\documentclass[a4paper,USenglish,cleveref, autoref, thm-restate]{lipics-v2019}

\usepackage{stmaryrd} 
\usepackage{xspace}
\usepackage{complexity}
\usepackage{bm}
\usepackage{tikz}
\usetikzlibrary{petri,positioning, fit, shapes, calc}

\newcommand{\defeq}{\stackrel{\scriptscriptstyle\text{def}}{=}}

\newcommand{\ie}{\text{i.e.}\xspace}

\newcommand{\N}{\mathbb{N}}                    
\renewcommand{\vec}[1]{\bm{#1}}                
\newcommand{\set}[1]{\left\{#1\right\}}        
\newcommand{\size}[1]{\mathrm{size}(#1)}       

\newcommand{\multiset}[1]{\Lbag#1\Rbag}        
\newcommand{\norm}[1]{\lVert#1\rVert}          

\newcommand{\support}[1]{\norm{#1}} 

\newcommand{\pre}{\mathit{pre}} 
\newcommand{\post}{\mathit{post}} 
\newcommand{\preset}[1]{{}^\bullet #1}  
\newcommand{\postset}[1]{{#1}^\bullet}  
\newcommand{\prestar}{\mathit{pre}^*}
\newcommand{\poststar}{\mathit{post}^*}

\newcommand{\cube}{\mathcal{C}}

\newcommand{\bio}{BIO} 
 
\newcommand{\bioLong}{branching immediate observation} 

\newcommand{\net}{N}
\newcommand{\trans}[1]{\xrightarrow{#1}}       
\newcommand{\placeCount}{n}

\newcommand{\firingSequence}{\sigma}

\newcommand{\targetMarking}{M}
\newcommand{\sourceMarking}{M'}

\newcommand{\agLength}[1]{|{#1} |_a} 

\newcommand{\sourceSize}{m'}
\newcommand{\targetSize}{m}

\newcommand{\dec}[1]{\widehat{#1}} 
\newcommand{\bunchreach}{\mathcal{R}_{p,i}}
\newcommand{\shortpath}{\rho}

\definecolor{lightred}{rgb}{1,0.8,0.8}

\title{Flatness and Complexity of Immediate Observation Petri Nets}

\author{Mikhail Raskin}{Technical University of Munich, Munich, Germany}{raskin@in.tum.de}{https://orcid.org/0000-0002-6660-5673}{}

\author{Chana Weil-Kennedy}{Technical University of Munich, Munich, Germany}{chana.weilkennedy@in.tum.de}{https://orcid.org/0000-0002-1351-8824}{}

\author{Javier Esparza}{Technical University of Munich, Munich, Germany}{esparza@in.tum.de}{https://orcid.org/0000-0001-9862-4919}{}

\authorrunning{M. Raskin, C. Weil-Kennedy, J. Esparza}

\Copyright{Mikhail Raskin and Chana Weil-Kennedy and Javier Esparza} 



\category{} 

\relatedversion{} 

\supplement{}

\funding{This project has received funding from the European Research Council (ERC) under the European Union's Horizon 2020 research and innovation programme under grant agreement No 787367 (PaVeS)}

\acknowledgements{We thank Jérôme Leroux and Rupak Majumdar for interesting conversations that put us on the path of flatness and BIO nets. We  also thank the reviewers whose comments allowed us to improve this paper, and fix a small mistake in Lemma 18.}

\nolinenumbers 

\hideLIPIcs  

\EventEditors{Igor Konnov and Laura Kov\'{a}cs}
\EventNoEds{2}
\EventLongTitle{31st International Conference on Concurrency Theory (CONCUR 2020)}
\EventShortTitle{CONCUR 2020}
\EventAcronym{CONCUR}
\EventYear{2020}
\EventDate{September 1--4, 2020}
\EventLocation{Vienna, Austria}
\EventLogo{}
\SeriesVolume{2017}
\ArticleNo{19}

\begin{document}

\maketitle 

\begin{abstract}
In a previous paper we introduced immediate observation (IO) Petri nets, a class of interest in the study of population protocols and enzymatic chemical networks. In the first part of this paper we show that IO nets are globally flat, and so their safety properties can be checked by efficient symbolic model checking tools using acceleration techniques,
like FAST. 
In the second part  we study Branching IO nets (BIO nets), whose transitions can create tokens. BIO nets extend both IO nets and communication-free nets, also called BPP nets, a widely studied class. We show that, while BIO nets are no longer globally flat, and their sets of reachable markings may be non-semilinear, they are still locally flat. As a consequence, the coverability and reachability problem for BIO nets, and even a certain set-parameterized version of them, are in PSPACE. This makes BIO nets the first natural net class with non-semilinear reachability relation for which the reachability problem is provably simpler than for general Petri nets.
\end{abstract}
%
\section{Introduction}
\label{sec:introduction}
Immediate observation Petri nets (IO nets) model immediate observation population protocols,
as introduced by Angluin \textit{et al.} in their seminal paper on the expressive power of population protocols \cite{journals/dc/AngluinAER07}.
In an IO net each transition is defined by three places: the source place $p_s$,
the destination place $p_d$, and the observed place $p_o$. The transition can move one token from $p_s$ to $p_d$, provided that $p_o$ is not empty (if $p_s=p_o$, then $p_o$ should contain at least two tokens). In the population protocol interpretation, $p_s$, $p_d$, and $p_o$ are three possible states of each of the identical agents executing the protocol, and a transition models an agent in the state $p_s$ observing another agent in the state $p_o$ and switching to the state $p_d$.

In a previous paper~\cite{conf/apn/EsparzaRW19}
 we investigated ``many-to-many'' versions of the reachability and coverability problems for IO nets,
 in which we have a set of initial markings and a set of final markings instead of the standard ``one-to-one'' versions with a single initial marking and a single final marking.
The sets we consider are \emph{cubes}, i.e., sets of markings obtained by attaching to each place a lower bound and an upper bound (possibly infinite) 
for the number of tokens. 
We showed that while the standard one-to-one problems are \PSPACE-hard, they remain in \PSPACE \ in the many-to-many case. 
This is in strong contrast with general conservative Petri nets (nets in which transitions neither create nor destroy tokens), for which 
the many-to-many versions of the problems become  \EXPSPACE-hard or even non-elementary.

In this paper we continue our study of IO nets, and initiate the study of Branching IO nets (BIO nets for short),
in which transitions can create or destroy agents. BIO nets deserve study for at least three reasons:
\begin{itemize}
\item They are a natural generalization of both IO nets and communication-free nets (aka BPP nets), another very well studied subclass (see e.g. \cite{conf/lics/ChristensenHM93,journals/fuin/Esparza97,DBLP:journals/tcs/Yen97,DBLP:conf/wflp/Fribourg00,DBLP:conf/atva/LerouxS05,DBLP:journals/ipl/Lasota09,DBLP:journals/fuin/MayrW15}). 
\item The reachability sets of BIO nets are not necessarily semilinear. In particular, Hopcroft and Pansiot's well-known example of a Petri net with a non-semilinear reachability set (see \cite{journals/tcs/HopcroftP79}) is a BIO net. The classes of unbounded Petri nets for which the reachability problem is demonstrably simpler than for arbitrary Petri nets, like BPP-nets, reversible nets, and IO nets, have semilinear reachability sets. This makes BIO nets ideal to investigate the existence of efficient verification techniques that do not depend on semilinearity.
\item BIO nets are a natural model for enzymatic catalytic reactions of the form $A + C \rightarrow C + B_1 + \cdots + B_n$ with more than one product. For example, catalase degrades hydrogen peroxide into water and oxygen, a reaction of the form  $A + C \rightarrow C + B_1 + B_2$ \cite{Cheli04}. Since IO nets have been used to model and analyze enzymatic reactions $A + C \rightarrow C + B$ (see \cite{angeli2007petri,baldan2010petri,marwan2011petri}), we expect our results to find a similar application.
\end{itemize}

In this paper we prove that IO nets are globally flat, in the sense of Leroux and Sutre \cite{DBLP:conf/atva/LerouxS05}.  In particular, this shows that their reachability relation is semilinear.  
Since the reachability relation of BIO nets is not semilinear, this result cannot extend to BIO nets. However,  we prove that they are locally $\pre^*$-flat, also in the sense of
\cite{DBLP:conf/atva/LerouxS05} \footnote{Actually, the locally flat of \cite{DBLP:conf/atva/LerouxS05}
are what we call locally $\post^*$-flat. A net is locally $\pre^*$-flat if{}f its reverse net is locally $\post^*$-flat,
and so with respect to reachability questions the difference is immaterial.}. Both global and local flatness allow us to analyze nets 
applying existing symbolic model checking tools like FAST \cite{conf/cav/BardinFLP03}, LASH \cite{boigelot2014toolset} and TREX \cite{conf/cav/AnnichiniBS01}. 
Further, we prove that 
the many-to-many versions of the reachability and coverability problems for BIO nets are still \PSPACE-complete, as for IO nets. To the best of our knowledge, this makes BIO nets the first natural class of nets whose 
reachability relation is non-semilinear for which these problems have elementary complexity.

Our flatness and complexity results are consequences of two theorems, called the Shortening Theorems for IO and BIO nets. They state that if $\targetMarking$ is reachable 
from $\sourceMarking$, then $\targetMarking$ can be reached by a sequence of bounded \emph{accelerated} length, defined as the length of the sequence after exhaustively 
replacing any subsequence of the form $tt$ by $t$. In the case of IO nets the accelerated length is independent of the initial and final markings, while for BIO nets it only depends on the final marking. 
We consider that the Shortening Theorems are also interesting in their own right. 

The paper is organized as follows. Section \ref{sec:preliminaries} contains preliminaries, and 
Section \ref{sec:io-bio} defines IO and BIO nets. Section \ref{sec:shortening} states the Shortening Theorems,
and derives our flatness and  complexity results for the one-to-one reachability and coverability problems 
as corollaries. The proof of the Shortening Theorem for IO nets is 
given in Section \ref{sec:proof-shortening-IO} and our main result, the Shortening Theorem for BIO nets, is proved 
in Section \ref{sec:proof-shortening-BIO}. Finally, we prove in Section \ref{sec:param-reach-BIO} that the
many-to-many reachability and coverability problems remain in \PSPACE.

\section{Preliminaries}
\label{sec:preliminaries}

\medskip \noindent \textbf{Multisets.} 
A \emph{multiset} on a finite set \(E\) is a mapping \(C \colon E \rightarrow \N\), i.e. for any $e\in E$, \(C(e)\) denotes the number of occurrences of element \(e\) in \(C\).
Let $\multiset{e_1,\ldots,e_n}$ denote the multiset $C$ such that $C(e)=|\{j\mid e_j=e\}|$.
Operations on \(\N\) like addition or comparison are extended to multisets by defining them component wise on each element of \(E\).
Subtraction is allowed in the following way: if $C,D$ are multisets on set $E$ then for all $e\in E$, $(C-D)(e)=\max (C(e)-D(e),0)$.
We call $|C| \defeq\sum_{e\in E} C(e)$ the \emph{size} of $C$, and $\support{C} \defeq \{ e \mid C(e)>0 \}$ the \emph{support} of $C$. 
Given a total order $e_1 \prec e_2 \prec \cdots \prec e_n$ on $E$, a multiset $C$ can be 
equivalently represented by the vector $(C(e_1), \ldots, C(e_n))\in \N^n$. 
A set $V \subseteq \N^n$ is \emph{linear} if there is a root $\vec{r} \in \N^n$ and a set $\{\vec{p}_1, \ldots, \vec{p}_n\}$ of periods such that
$V = \{ v + \sum_{i=1}^n \lambda_i \vec{p}_i \mid \lambda_1, \ldots, \lambda_n \in \N \}$, and \emph{semilinear} if it is the union of a finite set of linear sets.
A relation on $\N^n$ is semilinear if it is semilinear as a set of $\N^{2n}$. All these notions extend to sets of multisets.

\medskip \noindent \textbf{Place/transition Petri nets with weighted arcs.}
A \emph{Petri net} $N$ is a triple $(P,T,F)$ consisting of a finite set of \emph{places} $P$, a finite set of \emph{transitions} $T$ and a \emph{flow function} $F \colon (P \times T) \cup (T \times P) \rightarrow \mathbb{N}$. 
A \emph{marking} $M$ is a multiset on $P$, and we say that a marking $M$ puts $M(p)$ \emph{tokens} in place $p$ of $P$. The \emph{size} of $M$, denoted by $|M|$, is the total number of tokens in $M$.
The \emph{preset} $\preset{t}$ and \emph{postset} $\postset{t}$ of a transition $t$ are the multisets on $P$ given by $\preset{t}(p)=F(p,t)$ and $\postset{t}(p)=F(t,p)$. A transition $t$ is \emph{enabled} at a marking $M$ if $\preset{t} \leq M$, i.e. $\preset{t}$ is component-wise smaller or equal to $M$.
If $t$ is enabled then it can be \emph{fired}, leading to a new marking $M'=M - \preset{t} + \postset{t}$. 
We let $M \xrightarrow{t} M'$ denote this.

\medskip \noindent \textbf{Reachability and coverability.}
Given $\sigma=t_1 \ldots t_n$ we write $M \xrightarrow{\sigma} M_n$ when $M \xrightarrow{t_1} M_1 \xrightarrow{t_2} M_2 \ldots \xrightarrow{t_n} M_n$, and call $\sigma$ a \emph{firing sequence}. 
We write $M' \trans{*} M''$ if $M' \xrightarrow{\sigma} M''$ for some $\sigma \in T^*$, and say that $M''$ is \emph{reachable} from $M'$. 
A marking $M$ \emph{covers} another marking $M'$, written $M \geq M'$ if $M(p) \geq M'(p)$ for all places $p$.  
A marking $M$ is \emph{coverable} from $M'$ if there exists a marking $M''$ such that $M' \trans{*} M'' \geq M$.
The \emph{reachability relation} is the set of pairs of markings $(M,M')$ such that $M \trans{*} M'$, and we denote it $\trans{*}$.
The sets of predecessors and successors of a set $\mathcal{M}$ of markings of $\net$ are 
$\pre^*(\mathcal{M}) \defeq \{ M' | \exists M \in \mathcal{M} \, . \, M' \xrightarrow{*} M \}$ and 
$\post^*(\mathcal{M}) \defeq \{ M | \exists M' \in \mathcal{M} \, . \, M' \xrightarrow{*} M \}$, respectively.

\medskip \noindent \textbf{Global and local flatness.} 
A net $\net=(P,T,F)$ is \emph{globally flat} if there exist transition words $w_1,w_2, \ldots, w_k \in T^*$ such that for every two  
markings $\sourceMarking, \targetMarking$, if $\sourceMarking \trans{*} \targetMarking$, then there exist $j_1,\ldots,j_k\geq 0$ 
satisfying $\sourceMarking\trans{w_1^{j_1}\ldots w_k^{j_k}}\targetMarking$.
Observe that the words $w_1,w_2, \ldots, w_k$ are independent of both $M$ and $M'$. 
A net $\net=(P,T,F)$ $\net$ is \emph{locally $\pre^*$-flat} 
(resp. \emph{locally $\post^*$-flat})
if for every $\targetMarking$  (resp. $\sourceMarking$)
there exist transition words $w_1,w_2, \ldots, w_k \in T^*$ such that for every $\sourceMarking$
(resp. $\targetMarking$) satisfying $\sourceMarking \trans{*} \targetMarking$ 
there exist $j_1,\ldots,j_k\geq 0$ such that $\sourceMarking\trans{w_1^{j_1}\ldots w_k^{j_k}}\targetMarking$.
The locally flat Petri nets of \cite{DBLP:conf/atva/LerouxS05} correspond to our $\post^*$-flat nets.

\section{Immediate Observation and Branching Immediate Observation Nets}
\label{sec:io-bio}
We recall the definition of immediate observation nets (IO nets), as introduced in \cite{conf/apn/EsparzaRW19},
and extend it to \bioLong{} nets (\bio{} nets).

\begin{definition} 
\label{def:IOnet}
A transition $t$ of a Petri net is an \emph{immediate observation transition} (IO transition)  if there are 
places $p_s, p_d, p_o$, not necessarily distinct, such that $\preset{t}=\multiset{p_s,p_o}$ and $\postset{t}=\multiset{p_d,p_o}$.
We call $p_s, p_d, p_o$ the \emph{source}, \emph{destination}, and \emph{observed} places of $t$, respectively. 
A Petri net is an \emph{immediate observation net} (IO net) 
if all its transitions are IO transitions.

A transition $t$ of a Petri net is a \emph{branching IO transition} (\bio{} transition)  if there is $k \geq 0$ and 
places $p_s, p_{d_1}, \ldots, p_{d_k}, p_o$, not necessarily distinct, such that $\preset{t}=\multiset{p_s,p_o}$ and 
$\postset{t}=\multiset{p_{d_1}, \ldots, p_{d_k},p_o}$. 
A Petri net is a \emph{branching IO net} (\bio{} net)  if all its transitions are \bio{} transitions. 
\end{definition}

In the following examples, we allow ourselves to consider IO and BIO nets containing transitions with no observed place. 
To make the net a formally correct IO or BIO net, it suffices to add an extra marked place which acts as observed place for these transitions.


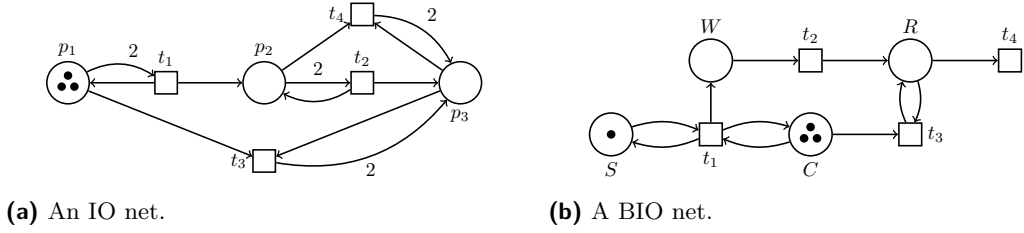
\begin{figure}[t]
\centering%
\begin{subfigure}[t]{0.49\textwidth}
\centering%
\resizebox{6cm}{!}{
\begin{tikzpicture}[->, node distance=1.75cm, auto, thick]
      \node[place, tokens=3] (p1) {};
      \node[transition] (t1) [right of=p1] {};
      \node[place] (p2) [right of=t1] {};
      \node[transition] (t3) [below=0.8 of p2] {};
      \node[transition] (t2) [right of=p2] {};
      \node[place] (p3) [right of=t2] {};
      \node[transition] (t4) [above=0.8 of t2] {};
      
      \path[->]
      (p1) edge[bend left] node {$2$} (t1)
      (t1) edge node {} (p1)
      (t1) edge node {} (p2)
      (p2) edge node[above] {$2$} (t2)
      (t2) edge[bend left] node {} (p2)     
      (t2) edge node {} (p3)
      (p1) edge node {} (t3)
      (p2) edge node {} (t4)
      (p3) edge node {} (t4)
      (t4) edge[bend left] node {$2$} (p3)
      (p3) edge node {} (t3)
      (t3) edge[bend right] node[below] {$2$} (p3)
      ;

      \node[] () [above= -1pt of p1] {$p_1$};
      \node[] () [above= -1pt of p2] {$p_2$};
      \node[] () [below= -1pt of p3] {$p_3$};
      \node[] () [above= -1pt of t1] {$t_1$};
      \node[] () [above= -1pt of t2] {$t_2$};
      \node[] () [left= -1pt of t3] {$t_3$};
      \node[] () [left= -1pt of t4] {$t_4$};
\end{tikzpicture}
}%
\caption{An IO net.}%
\label{fig:io}
\end{subfigure}\hfill%
\begin{subfigure}[t]{0.49\textwidth}
\centering%
\resizebox{6cm}{!}{
    \begin{tikzpicture}[->, node distance=1.75cm, auto, thick]
      \node[place] (p1) {};
      \node[transition] (t1) [right of=p1] {};
      \node[place] (p2) [right of=t1] {};
      \node[transition] (t2) [right of=p2] {};
      \node[place, tokens=3] (p4) [below=0.7 of t1] {};
      \node[transition] (t3) [below=0.7 of p1] {};
      \node[place, tokens=1] (p3) [left of=t3] {};
      \node[transition] (t4) [below=0.7 of p2] {};
      
      \path[->]
      (p1) edge node {} (t1)
      (t1) edge node {} (p2)
      (p4) edge node {} (t4)
      (p2) edge node {} (t2)
      (t3) edge node {} (p1)
      (p2) edge[bend left=20] node {} (t4)
      (t4) edge[bend left=20] node {} (p2)
      (p4) edge[bend left=20] node {} (t3)
      (t3) edge[bend left=20] node {} (p4)
      (p3) edge[bend left=20] node {} (t3)
      (t3) edge[bend left=20] node {} (p3)
      ;

      \node[] () [above= -1pt of p1] {$W$};
      \node[] () [above= -1pt of p2] {$R$};
      \node[] () [below= -1pt of p3] {$S$};
      \node[] () [below= -1pt of p4] {$C$};
      \node[] () [above= -1pt of t1] {$t_2$};
      \node[] () [above= -1pt of t2] {$t_4$};
      \node[] () [below= -1pt of t3] {$t_1$};
      \node[] () [right= -1pt of t4] {$t_3$};
    \end{tikzpicture}
}%
\caption{A \bio{} net.}%
\label{fig:bimo-example}
\end{subfigure}%
\caption{Examples of IO and \bio{} nets.}%
\end{figure}
\begin{example}
Figure \ref{fig:io} shows  an IO net taken from the literature on population protocols \cite{journals/dc/AngluinAER07}. 
Intuitively, it models a protocol allowing a crowd of undistinguishable agents that can only interact in pairs 
to decide whether they are at least 3. Initially all agents are in state $p_1$, modelled by tokens in place $p_1$. 
If two agents in state $p_1$ interact, one of them moves to state $p_2$ (transition $t_1$). 
If two agents in state $p_2$ interact, one 
of them moves to $p_3$ (transition $t_2$). 
Finally, an agent in state $p_3$ can ``attract'' all other agents to state $p_3$
(transitions $t_3$ and $t_4$). 
Given a marking $M_0$ with tokens only in $p_1$, if $M_0(p_1) \geq 3$ and the pairs 
of tokens that interact next are chosen uniformly at random, then eventually all tokens reach $p_3$.

Figure \ref{fig:bimo-example} shows a BIO net representing a client server interaction.
If the server $S$ observes a client $C$, it creates a worker $W$, which
creates a response $R$ and terminates. The client $C$ ``leaves'' after observing a response.
Responses may expire.
\end{example}


IO nets are \emph{conservative}, \ie there is no creation or destruction of tokens, while \bio{} nets are not.
The next example, taken from \cite{journals/tcs/HopcroftP79}, shows that BIO nets may have
non-semilinear sets of reachable markings.
\begin{example}[\cite{journals/tcs/HopcroftP79}]
\label{ex:non-semilinear-bio}
Consider the \bio{} net $\net$ of Figure \ref{non-semilinear-bimo}, with states $p,q,c_1,c_2,c_3$ and initial marking $M_0=(1,0,0,0,1)$.
The set of markings reachable from $M_0$ in $\net$ is characterized by the condition
$ (\vec{p}=1 \land \vec{q}=0 \land 0<\vec{c_2}+\vec{c_3}\leq 2^{\vec{c_1}})
\lor 
(\vec{p}=0 \land \vec{q}=1 \land 0<2\vec{c_2}+\vec{c_3}\leq 2^{\vec{c_1+1}})
$,
where $\vec{c}$ denotes the number of tokens in some place $c$.
Informally, one token cycles between $p$ and $q$, putting a new token in $c_1$ at every new cycle. 
When $p$ is marked, tokens in $c_3$ can move to $c_2$, and when $q$ is marked, tokens in $c_2$ can move to $c_3$ while doubling their number (see Lemma 2.8 of \cite{journals/tcs/HopcroftP79}).
Clearly the reachability relation of this \bio{} net is not semilinear.
\end{example}

\begin{figure}[ht]
\centering
\vskip-0.5cm
\resizebox{7cm}{!}{
    \begin{tikzpicture}[->, node distance=1.75cm, auto, thick]

     \node[place] (p1) {};
     \node[transition] (t1) [left of=p1] {};
     \node[transition] (t2) [left=1.6 of t1] {};
     \node[place, tokens=1] (p2) [above left=0.8 of t1] {};
     \node[place] (p4) [below left=0.8 of t1] {};
     \node[transition] (t3) [left=2 of p2] {};
     \node[transition] (t4) [left=2 of p4] {};
     \node[place] (p3) [left=2 of t3] {};
     \node[place, tokens=1] (p5) [left=2 of t4] {};
      
      \path[->]
      (t1) edge node {} (p1)
      (t1) edge node {} (p2)
      (p4) edge node {} (t1)
      (p2) edge node {} (t2)
      (t2) edge node {} (p4)
      (p5) edge node {} (t3)
      (t3) edge node {} (p3)
      (p2) edge[bend left=20] node {} (t3)
      (t3) edge[bend left=20] node {} (p2)
      (p3) edge node {} (t4)
      (t4) edge node {$2$} (p5)
      (p4) edge[bend left=20] node {} (t4)
      (t4) edge[bend left=20] node {} (p4)
      ;

      \node[] () [above= -1pt of p1] {$c_1$};
      \node[] () [right= -1pt of p2] {$p$};
      \node[] () [left= -1pt of p3] {$c_2$};
      \node[] () [right= -1pt of p4] {$q$};
      \node[] () [left= -1pt of p5] {$c_3$};
      \node[] () [left= -1pt of t1] {$t_4$};
      \node[] () [left= -1pt of t2] {$t_2$};
      \node[] () [below= -1pt of t3] {$t_1$};
      \node[] () [above= -1pt of t4] {$t_3$};
    \end{tikzpicture}}
\caption{A non-flat \bio{} net.}
\label{non-semilinear-bimo}
\vskip-0.5cm
\end{figure}

\section{Shortening Theorems}
\label{sec:shortening}
We introduce the main results of our paper, called the Shortening Theorems. We use them to prove flatness results for 
IO and BIO nets, and to extend complexity results of \cite{conf/apn/EsparzaRW19} for the reachability and coverability 
problems of IO nets to the (much harder) case of BIO nets. The Shortening Theorems  themselves are proved in Sections \ref{sec:proof-shortening-IO} and 
\ref{sec:proof-shortening-BIO}, respectively. 

First, we introduce a measure of the length of firing sequences that abstracts from
the number of times a transition is consecutively executed.

\begin{definition}
Let $\net$ be a Petri net, and let $\firingSequence$ be a firing sequence.
Let $k_1, \ldots, k_m$ be the unique positive natural numbers such that $\sigma=t_1^{k_1} t_2^{k_2} \ldots t_m^{k_m}$ and $t_i \neq t_{i+1}$ for every $i=1, \ldots, m-1$. We say that $\firingSequence$ has \emph{accelerated length} $m$, and let $\agLength{\firingSequence}$ denote the accelerated length of $\firingSequence$.
\end{definition}

The Shortening Theorems for IO and BIO show that a firing sequence leading from $M'$ to $M$ can be shortened to a sequence of bounded accelerated length. 
For IO nets the bound only depends on the net, not on the markings $M$ or $M'$:

\begin{restatable}[IO Shortening]{theorem}{ThmShortIO}
\label{thm:short-io}
Let $\net$ be an IO net with $\placeCount$ places, and let $\sourceMarking, \targetMarking$ be two markings of $\net$.
If $\sourceMarking \trans{*} \targetMarking$, then $\sourceMarking \trans{\firingSequence} \targetMarking$ for some $\sigma$
of accelerated length $\agLength{\firingSequence} \leq (\placeCount^3 +1)^{\placeCount}$.
\end{restatable}

Example \ref{ex:non-semilinear-bio} shows that for BIO nets the bound cannot be independent of both $M$ and $M'$:

\begin{example}
Recall the \bio{} net of Example \ref{ex:non-semilinear-bio} with states $p,q,c_1,c_2,c_3$.
It is easy to see that for $j\ge 1$ the marking $M_j \defeq (1,0, j, 0, 2^j)$
is reachable only via the firing sequence
$$(t_1t_2t_3t_4) (t_1^2t_2t_3^2t_4) \ldots (t_1^it_2t_3^it_4) \ldots (t_1^{j}t_2t_3^{j}t_4).$$
This sequence has accelerated length $4j$, which depends on the target marking $M_j$.
\end{example}

However, we can still obtain a bound independent of $M'$:

\begin{restatable}[\bio{} Shortening]{theorem}{BimoIntermediateValues}
\label{thm:bimo-intermediate-values}
Let $\net$ be a BIO net with $\placeCount$ places, let $\sourceMarking, \targetMarking$ be two markings of $\net$, and let
$|\sourceMarking|=\sourceSize$, $|\targetMarking|=\targetSize$.
Let $m_d := \max_{t \in T} |\postset{t} - \preset{t}|$ denote the maximum number of tokens created by a transition of $\net$.
If $\sourceMarking \trans{*} \targetMarking$, then $\sourceMarking \trans{\firingSequence} \targetMarking$ for some $\sigma$
of accelerated length
$\agLength{\firingSequence} \leq 2^n(m+1)^n(n+1)^n $. Further, the 
intermediate markings along $\firingSequence$ have size at most
$
(m' +  2^n(m+1)^n(n+1)^n  (m+n) m_d) m_d^n
$.
\end{restatable}

\subsection{Flatness and complexity results}
\label{sec:flatness}
The Shortening Theorems lead easily to our flatness and complexity results:

\begin{theorem}
\label{thm:io-flat}
IO nets are globally flat. BIO nets are locally $\prestar$-flat, but neither globally flat nor locally $\poststar$-flat.
\end{theorem}
\begin{proof}
\noindent (a) We show that IO nets are globally flat. 
Let $\net=(P,T,F)$ be an IO net with $n$ places and $T= \{t_1, \ldots, t_m\}$,
and let $K=(\placeCount^3 +1)^{\placeCount}$.
By Theorem \ref{thm:short-io}, for every two markings $\sourceMarking$ and $\targetMarking$ of $\net$ there is a 
firing sequence $t_{i_1}^{j_1} \cdots t_{i_K}^{j_K}$ leading from $\sourceMarking$ to $\targetMarking$. Since every such sequence belongs to the regular language
$ \left( t_1^*t_2^* \cdots t_m^*\right)^{K}$, the words $w_1,w_2, \ldots, w_{m \cdot K}$  given by $w_i = t_{((i-1) \, \text{mod} \, m)+1}$ for every $1 \leq i \leq m \cdot K$ witness that $\net$ is globally flat.

\smallskip

\noindent (b) We show that BIO nets are locally $\prestar$-flat. Let $\net=(P,T,F)$ be a BIO net with $n$ places and $T= \{t_1, \ldots, t_m\}$, let $\targetMarking$ be a marking of $\net$ with $|\targetMarking|=\targetSize$, and let $K = 2^n(m+1)^n(n+1)^n$.
By Theorem \ref{thm:bimo-intermediate-values}, for every marking $\sourceMarking$ of $\net$ there is a 
firing sequence $t_{i_1}^{j_1} \cdots t_{i_K}^{j_K}$ leading from $\sourceMarking$ to $\targetMarking$. Proceed now as for (a).

\smallskip

\noindent (c) We show that BIO nets are not locally $\poststar$-flat, and so also not globally flat. 
Consider the \bio{} net of Figure \ref{non-semilinear-bimo} with states $p,q,c_1,c_2,c_3$. Recall
that for all $j \ge 1$, $M_0$ only reaches the marking $M_j \defeq (1,0, j, 0, 2^j)$ via 
$(t_1t_2t_3t_4) (t_1^2t_2t_3^2t_4) \ldots (t_1^it_2t_3^it_4) \ldots (t_1^{j}t_2t_3^{j}t_4)$.
So in order to reach $M_j$ it is necessary to fire $j$ times a sequence of the form $t_1^{k}t_2^{k_2}t_3^{k}t_4^{k_4}$, which proves the result.
\end{proof}

\begin{theorem}
\label{thm:simple-BIO-reachability}
The reachability and coverability problems for \bio{} nets are \PSPACE-complete.
\end{theorem}
\begin{proof}
Reachability and coverability are \PSPACE-complete for IO nets \cite{conf/apn/EsparzaRW19}, and IO nets are a subclass of \bio{} nets, so the problems stay \PSPACE-hard for \bio{} nets. 
By Savitch's theorem it suffices to show that the problems are in \NPSPACE.
Consider first the reachability problem. By the Shortening Theorem, given a \bio{} net with $n$ places and two markings $M$ and $M'$ we can guess a firing sequence leading from $M$ to $M'$, if one
exists, using space $\log (f(n,m,m',m_d))$, where $f(n,m,m',m_d)$ is the exponential bound of the Shortening Theorem. So the reachability
problem is in \NPSPACE.  For coverability, we reduce it to reachability in the usual way. Let $\targetMarking$ be the marking we want to cover.
For each place $p$, we add a ``destroying transition'' $\tau_p$ with preset $\preset{t}=\set{p}$ and postset $\postset{t}=\emptyset$.
It is easy to see that for every marking $\sourceMarking$, the modified net $\net'$ has a firing sequence from $\sourceMarking$ to $\targetMarking$
i{}ff $\net$ has a firing sequence from $\sourceMarking$ to some marking covering $\targetMarking$.
\end{proof}

\section{Shortening Theorem for IO nets}
\label{sec:proof-shortening-IO}
The proof of Theorem \ref{thm:short-io} is based on a result of \cite{conf/apn/EsparzaRW19} called the Pruning Lemma.
We briefly introduce some notions required to state the lemma, and then the lemma itself. More details
can be found in \cite{conf/apn/EsparzaRW19}.

\subparagraph*{Trajectories and histories.} Since the transitions of IO nets do not create or destroy tokens, we can give tokens identities. 
Given a firing sequence, each token of the initial marking follows a \emph{trajectory} through the places of the net until it 
reaches the final marking of the sequence. The trajectories of the tokens between given source
and target markings constitute a \emph{history}. 

Fix an IO net $N$. A \emph{trajectory} of an IO net $N$ is a sequence $\tau =p_1 \ldots p_k$ of places. We let  $\tau(i)$ denote the $i$-th place of $\tau$. The \emph{$i$-th step} of $\tau$ is the pair $\tau(i)\tau(i+1)$. A \emph{history} $H$ of length $h$ is a multiset of trajectories of length $h$. Given an index $1 \leq i \leq h$, \emph{the $i$-th marking of $H$}, denoted $M_{H}^i$, is defined as follows: for every place $p$, $M_{H}^i(p)$ is the number of trajectories $\tau \in H$ such that $\tau(i)=p$. The markings 
$M_{H}^1$ and $M_{H}^h$ are the \emph{initial} and \emph{final} markings of $H$, and we write
$M_{H}^1 \trans{H} M_{H}^h$. A history $H$ of length $h\geq 1$ is \emph{realizable} if there exist transitions $t_1, \ldots, t_{h-1}$ 
and numbers $k_1, \ldots, k_{h-1} \geq 0$ such that 
\begin{itemize}
\item $M_{H}^1 \trans{t_1^{k_1}}M_{H}^2 \cdots  M_{H}^{h-1} \trans{t_{h-1}^{k_{h-1}}} M_{H}^h$, where for every $t$  we define $\sourceMarking \trans{t^0} \targetMarking$ if{}f $\sourceMarking=\targetMarking$.
\item For every $1 \leq i \leq h-1$, there are exactly $k_i$ trajectories $\tau \in H$ such that $\tau(i)\tau(i+1) = p_s p_d$, where $p_s, p_d$ are the source and target places of $t_i$, and all other trajectories 
$\tau \in H$ satisfy $\tau(i)=\tau(i+1)$.
Moreover, there is at least one trajectory $\tau$ in $H$ such that $\tau(i)\tau(i+1) = p_o p_o$, where $p_o$ is the observed place of $t_i$.
\end{itemize}
We say that $t_1^{k_1} \cdots t_{h-1}^{k_{h-1}}$ realizes $H$.
Intuitively, at a step of a realizable history only one transition occurs, although perhaps multiple times, for
different tokens.  From the definition of realizable history we immediately obtain:
\begin{itemize}
\item $\sourceMarking \trans{*} \targetMarking$ if{}f there exists a realizable history with $\sourceMarking$ and $\targetMarking$ as initial and final markings. 
\item Every firing sequence that realizes a history of length $h$ has accelerated length at most $h$.
\end{itemize}

\newcommand{\distIO}{*0.7}

\begin{figure}[t]
\centering%
\begin{subfigure}[t]{0.49\textwidth}
\centering%
\resizebox{7cm}{!}{
\begin{tikzpicture}[scale = 0.4]
  \node[circle, draw]
    (Node11)
    {~~~}; 
    \node[right=1\distIO of Node11, circle, draw]
    (Node12)
    {~~~};
    \node[right=1\distIO of Node12, circle, draw]
    (Node13)
    {~~~};
    \node[right=1\distIO of Node13, circle, draw]
    (Node14)
    {~~~};
    \node[right=1\distIO of Node14, circle, draw]
    (Node15)
    {~~~};
    \node[right=1\distIO of Node15, circle, draw]
    (Node16)
    {~~~};
    \node[right=1\distIO of Node16, circle, draw]
    (Node17)
    {~~~};
    \node[below=0.7\distIO of Node11, circle, draw]
    (Node21)
    {~~~};
    \node[right=1\distIO of Node21, circle, draw]
    (Node22)
    {~~~};
    \node[right=1\distIO of Node22, circle, draw]
    (Node23)
    {~~~};
    \node[right=1\distIO of Node23, circle, draw]
    (Node24)
    {~~~};
    \node[right=1\distIO of Node24, circle, draw]
    (Node25)
    {~~~};
    \node[right=1\distIO of Node25, circle, draw]
    (Node26)
    {~~~};
    \node[right=1\distIO of Node26, circle, draw]
    (Node27)
    {~~~};
    \node[below=0.7\distIO of Node21, circle, draw]
    (Node31)
    {~~~};
    \node[right=1\distIO of Node31, circle, draw]
    (Node32)
    {~~~};
    \node[right=1\distIO of Node32, circle, draw]
    (Node33)
    {~~~}; 
    \node[right=1\distIO of Node33, circle, draw]
    (Node34)
    {~~~};
    \node[right=1\distIO of Node34, circle, draw]
    (Node35)
    {~~~};
    \node[right=1\distIO of Node35, circle, draw]
    (Node36)
    {~~~};
    \node[right=1\distIO of Node36, circle, draw]
    (Node37)
    {~~~};
        \node[left=0.1 of Node11]
        (Node11l)
        {$p_1$};
        \node[left=0.1 of Node21]
        (Node21l)
        {$p_2$};
        \node[left=0.1 of Node31]
        (Node31l)
        {$p_3$};

\def\d{3}
  \draw
  ($(Node11.center)+(0,1.50mm*\d)$)
    --
  ($(Node12.center)+(0,1.50mm*\d)$)
  ;
  \draw
  ($(Node12.center)+(0,1.50mm*\d)$)
    --
  ($(Node13.center)+(0,1.50mm*\d)$)
  ;
  \draw($(Node13.center)+(0,1.50mm*\d)$)
    --
  ($(Node14.center)+(0,1.50mm*\d)$)
  ;
  \draw
  ($(Node14.center)+(0,1.50mm*\d)$)
    --
  ($(Node15.center)+(0,1.50mm*\d)$)
  ;
  \draw
  ($(Node15.center)+(0,1.50mm*\d)$)
    --
  ($(Node16.center)+(0,1.50mm*\d)$)
  ;
  \draw
  ($(Node16.center)+(0,1.50mm*\d)$)
    --
  ($(Node17.center)+(0,1.50mm*\d)$)
  ;
    
  \draw
  ($(Node11.center)+(0,0.75mm*\d)$)
    --
  ($(Node12.center)+(0,0.75mm*\d)$)
  ;
  \draw
  ($(Node12.center)+(0,0.75mm*\d)$)
    --
  ($(Node13.center)+(0,0.75mm*\d)$)
  ;
  \draw($(Node13.center)+(0,0.75mm*\d)$)
    --
  ($(Node14.center)+(0,0.75mm*\d)$)
  ;
  \draw
  ($(Node14.center)+(0,0.75mm*\d)$)
    --
  ($(Node35.center)+(0,0.00mm*\d)$)
  ;
  \draw
  ($(Node35.center)+(0,0.00mm*\d)$)
    --
  ($(Node36.center)+(0,0.00mm*\d)$)
  ;
  \draw
  ($(Node36.center)+(0,0.00mm*\d)$)
    --
  ($(Node37.center)+(0,0.00mm*\d)$)
  ;
  \draw
  ($(Node11.center)+(0,-0.75mm*\d)$)
    --
  ($(Node12.center)+(0,-0.75mm*\d)$)
  ;
  \draw
  ($(Node12.center)+(0,-0.75mm*\d)$)
    --
  ($(Node23.center)+(0,-0.75mm*\d)$)
  ;
  \draw
  ($(Node23.center)+(0,-0.75mm*\d)$)
    --
  ($(Node24.center)+(0,-0.75mm*\d)$)
  ;
  \draw
  ($(Node24.center)+(0,-0.75mm*\d)$)
    --
  ($(Node25.center)+(0,-0.75mm*\d)$)
  ;
  \draw
  ($(Node25.center)+(0,-0.75mm*\d)$)
    --
  ($(Node36.center)+(0,0.75mm*\d)$)
  ;
  \draw
  ($(Node36.center)+(0,0.75mm*\d)$)
    --
  ($(Node37.center)+(0,0.75mm*\d)$)
  ;
  \draw
  ($(Node11.center)+(0,0.0mm*\d)$)
    --
  ($(Node12.center)+(0,0.0mm*\d)$)
  ;
  \draw
  ($(Node12.center)+(0,0.0mm*\d)$)
    --
  ($(Node13.center)+(0,0.0mm*\d)$)
  ;
  \draw
  ($(Node13.center)+(0,0.0mm*\d)$)
    --
  ($(Node24.center)+(0,0.0mm*\d)$)
  ;
  \draw
  ($(Node24.center)+(0,0.0mm*\d)$)
    --
  ($(Node25.center)+(0,0.0mm*\d)$)
  ;
  \draw
  ($(Node25.center)+(0,0.0mm*\d)$)
    --
  ($(Node26.center)+(0,0.0mm*\d)$)
  ;
  \draw
  ($(Node26.center)+(0,0.0mm*\d)$)
    --
  ($(Node37.center)+(0,1.50mm*\d)$)
  ;
  \draw
  ($(Node11.center)+(0,-1.50mm*\d)$)
    --
  ($(Node32.center)+(0,-0.75mm*\d)$)
  ;
  \draw
  ($(Node32.center)+(0,-0.75mm*\d)$)
    --
  ($(Node33.center)+(0,-0.75mm*\d)$)
  ;
  \draw
  ($(Node33.center)+(0,-0.75mm*\d)$)
    --
  ($(Node34.center)+(0,-0.75mm*\d)$)
  ;
  \draw
  ($(Node34.center)+(0,-0.75mm*\d)$)
    --
  ($(Node35.center)+(0,-0.75mm*\d)$)
  ;
  \draw
  ($(Node35.center)+(0,-0.75mm*\d)$)
    --
  ($(Node36.center)+(0,-0.75mm*\d)$)
  ;
  \draw
  ($(Node36.center)+(0,-0.75mm*\d)$)
    --
  ($(Node37.center)+(0,-0.75mm*\d)$)
  ;
  \draw
  ($(Node31.center)+(0,-1.50mm*\d)$)
    --
  ($(Node32.center)+(0,-1.50mm*\d)$)
  ;
  \draw
  ($(Node32.center)+(0,-1.50mm*\d)$)
    --
  ($(Node33.center)+(0,-1.50mm*\d)$)
  ;
  \draw
  ($(Node33.center)+(0,-1.50mm*\d)$)
    --
  ($(Node34.center)+(0,-1.50mm*\d)$)
  ;
  \draw
  ($(Node34.center)+(0,-1.50mm*\d)$)
    --
  ($(Node35.center)+(0,-1.50mm*\d)$)
  ;
  \draw
  ($(Node35.center)+(0,-1.50mm*\d)$)
    --
  ($(Node36.center)+(0,-1.50mm*\d)$)
  ;
  \draw
  ($(Node36.center)+(0,-1.50mm*\d)$)
    --
  ($(Node37.center)+(0,-1.50mm*\d)$)
  ;

\end{tikzpicture}}
\caption{A realizable history.}
\label{figure-history}
\end{subfigure}%
\hfill%
\begin{subfigure}[t]{0.49\textwidth}
\centering%
\resizebox{7cm}{!}{
\begin{tikzpicture}[scale = 0.4]
  \node[circle, draw]
    (Node11)
    {~~~}; 
    \node[right=1\distIO of Node11, circle, draw]
    (Node12)
    {~~~};
    \node[right=1\distIO of Node12, circle, draw]
    (Node13)
    {~~~};
    \node[right=1\distIO of Node13, circle, draw]
    (Node14)
    {~~~};
    \node[right=1\distIO of Node14, circle, draw]
    (Node15)
    {~~~};
    \node[right=1\distIO of Node15, circle, draw]
    (Node16)
    {~~~};
    \node[right=1\distIO of Node16, circle, draw]
    (Node17)
    {~~~};
    \node[below=0.7\distIO of Node11, circle, draw]
    (Node21)
    {~~~};
    \node[right=1\distIO of Node21, circle, draw]
    (Node22)
    {~~~};
    \node[right=1\distIO of Node22, circle, draw]
    (Node23)
    {~~~};
    \node[right=1\distIO of Node23, circle, draw]
    (Node24)
    {~~~};
    \node[right=1\distIO of Node24, circle, draw]
    (Node25)
    {~~~};
    \node[right=1\distIO of Node25, circle, draw]
    (Node26)
    {~~~};
    \node[right=1\distIO of Node26, circle, draw]
    (Node27)
    {~~~};
    \node[below=0.7\distIO of Node21, circle, draw]
    (Node31)
    {~~~};
    \node[right=1\distIO of Node31, circle, draw]
    (Node32)
    {~~~};
    \node[right=1\distIO of Node32, circle, draw]
    (Node33)
    {~~~}; 
    \node[right=1\distIO of Node33, circle, draw]
    (Node34)
    {~~~};
    \node[right=1\distIO of Node34, circle, draw]
    (Node35)
    {~~~};
    \node[right=1\distIO of Node35, circle, draw]
    (Node36)
    {~~~};
    \node[right=1\distIO of Node36, circle, draw]
    (Node37)
    {~~~};
        \node[left=0.1 of Node11]
        (Node11l)
        {$p_1$};
        \node[left=0.1 of Node21]
        (Node21l)
        {$p_2$};
        \node[left=0.1 of Node31]
        (Node31l)
        {$p_3$};


\def\d{3}
  \draw
  ($(Node11.center)+(0,1.50mm*\d)$)
    --
  ($(Node12.center)+(0,1.50mm*\d)$)
  ;
  \draw
  ($(Node12.center)+(0,1.50mm*\d)$)
    --
  ($(Node13.center)+(0,1.50mm*\d)$)
  ;
  \draw($(Node13.center)+(0,1.50mm*\d)$)
    --
  ($(Node14.center)+(0,1.50mm*\d)$)
  ;
  \draw
  ($(Node14.center)+(0,1.50mm*\d)$)
    --
  ($(Node15.center)+(0,1.50mm*\d)$)
  ;
  \draw
  ($(Node15.center)+(0,1.50mm*\d)$)
    --
  ($(Node16.center)+(0,1.50mm*\d)$)
  ;
  \draw
  ($(Node16.center)+(0,1.50mm*\d)$)
    --
  ($(Node17.center)+(0,1.50mm*\d)$)
  ; 
        
  \draw[dashed]
  ($(Node11.center)+(0,0.75mm*\d)$)
    --
  ($(Node12.center)+(0,0.75mm*\d)$)
  ;
  \draw[dashed]
  ($(Node12.center)+(0,0.75mm*\d)$)
    --
  ($(Node13.center)+(0,0.75mm*\d)$)
  ;
  \draw[dashed]
  ($(Node13.center)+(0,0.75mm*\d)$)
    --
  ($(Node14.center)+(0,0.75mm*\d)$)
  ;
  \draw[dashed]
  ($(Node14.center)+(0,0.75mm*\d)$)
    --
  ($(Node35.center)+(0,0.00mm*\d)$)
  ;
  \draw[dashed]
  ($(Node35.center)+(0,0.00mm*\d)$)
    --
  ($(Node36.center)+(0,0.00mm*\d)$)
  ;
  \draw[dashed]
  ($(Node36.center)+(0,0.00mm*\d)$)
    --
  ($(Node37.center)+(0,0.00mm*\d)$)
  ;

  \draw[dashed]
  ($(Node11.center)+(0,0.0mm*\d)$)
    --
  ($(Node12.center)+(0,0.0mm*\d)$)
  ;
  \draw[dashed]
  ($(Node12.center)+(0,0.0mm*\d)$)
    --
  ($(Node23.center)+(0,0.0mm*\d)$)
  ;
  \draw[dashed]
  ($(Node23.center)+(0,0.0mm*\d)$)
    --
  ($(Node24.center)+(0,0.0mm*\d)$)
  ;
  \draw[dashed]
  ($(Node24.center)+(0,0.0mm*\d)$)
    --
  ($(Node25.center)+(0,0.0mm*\d)$)
  ;
  \draw[dashed]
  ($(Node25.center)+(0,0.0mm*\d)$)
    --
  ($(Node26.center)+(0,0.0mm*\d)$)
  ;
  \draw[dashed]
  ($(Node26.center)+(0,0.0mm*\d)$)
    --
  ($(Node37.center)+(0,1.50mm*\d)$)
  ;
  \draw[dashed]
  ($(Node11.center)+(0,-1.50mm*\d)$)
    --
  ($(Node32.center)+(0,-0.75mm*\d)$)
  ;
  \draw[dashed]
  ($(Node32.center)+(0,-0.75mm*\d)$)
    --
  ($(Node33.center)+(0,-0.75mm*\d)$)
  ;
  \draw[dashed]
  ($(Node33.center)+(0,-0.75mm*\d)$)
    --
  ($(Node34.center)+(0,-0.75mm*\d)$)
  ;
  \draw[dashed]
  ($(Node34.center)+(0,-0.75mm*\d)$)
    --
  ($(Node35.center)+(0,-0.75mm*\d)$)
  ;
  \draw[dashed]
  ($(Node35.center)+(0,-0.75mm*\d)$)
    --
  ($(Node36.center)+(0,-0.75mm*\d)$)
  ;
  \draw[dashed]
  ($(Node36.center)+(0,-0.75mm*\d)$)
    --
  ($(Node37.center)+(0,-0.75mm*\d)$)
  ;
  \draw
  ($(Node31.center)+(0,-1.50mm*\d)$)
    --
  ($(Node32.center)+(0,-1.50mm*\d)$)
  ;
  \draw
  ($(Node32.center)+(0,-1.50mm*\d)$)
    --
  ($(Node33.center)+(0,-1.50mm*\d)$)
  ;
  \draw
  ($(Node33.center)+(0,-1.50mm*\d)$)
    --
  ($(Node34.center)+(0,-1.50mm*\d)$)
  ;
  \draw
  ($(Node34.center)+(0,-1.50mm*\d)$)
    --
  ($(Node35.center)+(0,-1.50mm*\d)$)
  ;
  \draw
  ($(Node35.center)+(0,-1.50mm*\d)$)
    --
  ($(Node36.center)+(0,-1.50mm*\d)$)
  ;
  \draw
  ($(Node36.center)+(0,-1.50mm*\d)$)
    --
  ($(Node37.center)+(0,-1.50mm*\d)$)
  ;

\end{tikzpicture}}%
\caption{Pruning of the history on the left.}
\label{pruned-history}
\end{subfigure}%
\caption{A realizable history of the IO net of Figure \ref{fig:io} before and after pruning.}
\end{figure}
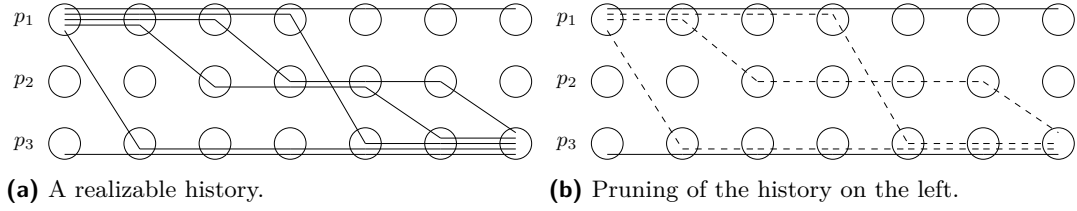

\begin{example}
\label{history-ex}
Figure \ref{figure-history} shows a realizable history of the IO net of Figure \ref{fig:io}.
It consists of six trajectories. The initial and final markings are $(5,0,1)$ and $(1,0,5)$. 
The history is realized by the firing sequence $t_3 t_1 t_1 t_3 t_2 t_4 $. 
\end{example}

\subparagraph*{Bunches and Pruning Lemma.} A \emph{bunch} is a multiset of trajectories with the same length and the same initial and final place. The Pruning Lemma states that every realizable history containing a bunch of trajectories from $p$ to $p'$ 
of size larger than the number of places $\placeCount$ can be ``pruned'',  meaning that the bunch can be replaced by a smaller one, 
also leading from $p$ to $p'$, while keeping the history realizable. (Notice, however, that the smaller bunch cannot always be chosen as a sub-multiset of the original one.) 

\begin{lemma}[Pruning Lemma]
\label{lm:pruning}
Let $N$ be an IO net with $\placeCount$ places. Let $H$ be a realizable history of $N$ containing a bunch $B\subseteq H$ of size larger than $\placeCount$. 
There exists a bunch $B'$ of size at most $\placeCount$ with the same initial and final
places as $B$, such that the history $H' \defeq H - B + B'$ (where $+$ and $-$ denote multiset addition and subtraction) is also realizable in $N$.
\end{lemma}

\begin{example}
\label{history-ex-pruned}
The realizable history $H$ of Figure \ref{figure-history}, leading from $(5,0,1)$ to $(1,0,5)$, has a bunch $B$ of size $4 \ge n$ from $p_1$ to $p_3$.
Figure \ref{pruned-history} shows a history $H'$, leading from $(4,0,1)$ to $(1,0,4)$,
 resulting from the application of the Pruning Lemma to $H$ and $B$.
The new bunch $B'$ from $p_1$ to $p_3$ given by the Pruning Lemma is drawn in dashed trajectories.
Notice that the trajectory of $B'$ that passes through $p_2$ does not appear in $B$.
The firing sequence $t_3 t_1 t_3 t_4 $ realizes $H'$.
\end{example}


\subparagraph*{Proof of the Shortening Theorem.} We need a Boosting Lemma, which states that duplicating 
a trajectory of a history of an IO net preserves realizability. Intuitively, duplicating a trajectory 
corresponds to adding a ``shadow'' to a token, that follows the token wherever it goes. Since an enabled IO transition can move 
arbitrarily many tokens from its source place to its destination place, the shadow token can always follow the primary token. 
A formal proof of the lemma is given in the Appendix.

\begin{restatable}[Boosting Lemma]{lemma}{BoostingLemma}
\label{lm:boosting}
Let $H$ be a realizable history of an IO net containing a trajectory $\tau$. 
The history $H + \multiset{\tau}$ is also realizable. 
\end{restatable}

\ThmShortIO*
\begin{proof}(Sketch.)
We explain our proof strategy for the IO Shortening Theorem. Given $M' \trans{*} M$, we take 
a history $H$ such that $M' \trans{H} M$. Repeatedly applying the Pruning Lemma, 
we construct another realizable history $\widetilde{H}$ such that $\widetilde{T}_{p, q} = \min\{ n, T_{p,q} \}$ 
for every two places $p$ and $q$, where $T_{p,q}$ and  $\widetilde{T}_{p, q}$ denote the 
number of trajectories of $H$ and $\widetilde{H}$ leading from $p$ to $q$. Using the fact that $\widetilde{H}$ has at most $n^3$ trajectories,
we show that $\widetilde{H}$ can be chosen so that its length is bounded by 
$(\placeCount^3 +1)^{\placeCount}$. We are not done yet, because in general $\widetilde{H}$ does not
lead from $M'$ to $M$, we only have  $\widetilde{M'} \trans{\widetilde{H}} \widetilde{M}$ for markings $\widetilde{M'}, \widetilde{M}$
such that $\widetilde{M'} \leq M'$ and $\widetilde{M} \leq M$. 
In the last step we use the Boosting Lemma to add trajectories to $\widetilde{H}$ without increasing 
its length, yielding a realizable history $\overline{H}$ of the same length as $\widetilde{H}$, but satisfying 
$M' \trans{\overline{H}} M$. Finally, we extract from $\overline{H}$ a sequence $M' \trans{\sigma} M$ of accelerated length
at most $(\placeCount^3 +1)^{\placeCount}$. The full proof can be found in the Appendix.
\end{proof}

\section{Shortening Theorem for BIO nets}
\label{sec:proof-shortening-BIO}
The proof of the BIO Shortening Theorem (Theorem \ref{thm:bimo-intermediate-values}) is very involved. It follows the
proof outline  of Theorem \ref{thm:short-io}:  Given a firing sequence, consider a history $H$ realized by it, construct an 
equivalent ``small'' history $H'$, and extract from $H'$ a sequence of short accelerated length. However, since BIO nets can create and destroy tokens,  trajectories must be generalized to branching trajectories, which are trees 
of places; intuitively, the tree captures the cascade of tokens created by a token of the initial marking. 

We fix a \bio{} net $\net=(P,T,F)$ with $n$ places, and let $m_d := \max_{t \in T} |\postset{t} - \preset{t}|$ denote the maximum number of tokens created by a transition.

\subparagraph*{Branching trajectories.}
A \emph{branching trajectory} of $N$ is a nonempty, directed tree $\beta$ whose nodes are labeled with places of $P$. A node labeled by $p$ is called a \emph{$p$-node}.
The $i$-th level of $\beta$, denoted by $\beta(i)$, is the (possibly empty) set of nodes of $\beta$ at distance $(i-1)$ from the root.  We let $M_\beta(i)$ denote the multiset of places labeling the nodes of $\beta(i)$. Observe that $M_\beta(i)$ is a marking.   We say that $\beta$ has \emph{length} $l$ if $\beta(l)\neq \emptyset$ and
$\beta(l+1)= \emptyset$.

\subparagraph*{Histories and realizable histories.} A \emph{history} $H$ of length $l$ is a forest of branching trajectories of length at most $l$. We use histories to describe a  behaviour from an initial marking; the history contains a branching trajectory for each token of the initial marking.  

Given a history $H$ of length $h$ and an index $1 \leq i \leq h$, the $i$-th level of $H$ is 
the set $H(i) = \bigcup_{\beta \in H} \beta(i)$, and the \emph{the $i$-th marking of $H$}, denoted $M_{H}^i$, is the 
multiset $M_{H}^i = \sum_{\beta \in H} M_\beta(i)$. The markings  $M_{H}^1$ and $M_{H}^h$ are called the 
\emph{initial} and \emph{final} markings of $H$, and we write $M_{H}^1 \trans{H} M_{H}^h$.
If the length of $H$ is longer than the length of its branching trajectories, the final marking of  $H$ is the zero marking. Two histories are \emph{equivalent} if they have the same initial and final markings.

A history $H$ of length $h\geq 1$ is \emph{realizable} if there exist transitions $t_1, \ldots, t_{h-1} \in T$ and numbers $k_1, \ldots, k_{h-1} \geq 0$ such that for every $1 \leq i \leq h-1$ the set $H(i)$ can be partitioned into two sets: 
\begin{itemize}
\item A set $H_{a}(i)$ of exactly $k_i$ nodes labeled by the source place of $t_i$.
We call these nodes \emph{active} nodes.
Given a particular active node, say $v$, the multiset of labels of its children is the (possibly empty) multiset $\multiset{p_{d_1}, \ldots, p_{d_k}}$ of destination places of $t_i$. 
\item A set $H_p(i)$ of nodes, each of them with exactly one child, carrying the same label as their parents. 
We call these nodes \emph{passive} nodes.
This set must contain at least one node labeled by the place $p_o$ observed by $t_i$.
\end{itemize}
\noindent We say that the sequence $t_1^{k_1} \cdots t_{h-1}^{k_{h-1}}$ \emph{realizes} $H$.
It follows easily from the definitions that $M_{H}^1 \trans{t_1^{k_1}}M_{H}^2 \cdots  M_{H}^{h-1} \trans{t_{h-1}^{k_{h-1}}} M_{H}^h$ holds (where $M \trans{t^0} M'$ if{}f $M=M'$).

From this definition we easily obtain:
\begin{itemize}
\item $M \trans{*} M'$ if{}f there exists a realizable history with $M$ and $M'$ as initial and final markings. 
\item Every firing sequence that realizes a history of length $h$ has accelerated length at most $h$.
\end{itemize}

\begin{figure}
\centering
\resizebox{13cm}{!}{
\input{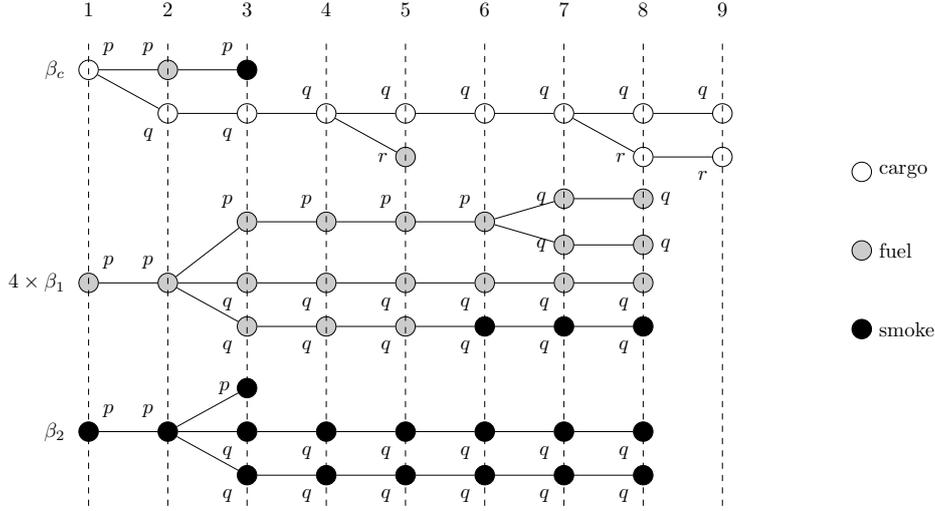}
}
\caption{A decorated realizable history of a BIO net.}
\label{fig:Bp(i)}
\end{figure}
\begin{example}-
\label{BIOhistory-ex}
Figure \ref{fig:Bp(i)} shows a realizable history $H$ of a \bio{} net with places $\{ p, q, r \}$.
$H$ consists of six branching trajectories: $\beta_{c}$, four copies of $\beta_1$, and $\beta_2$. 
The initial and final markings are $(6,0,0)$ and $(0,1,1)$.  The transition $t_i$ executed at step $i$  is
$$\begin{array}{llll}
t_1 = p \trans{p} \multiset{q,p} & t_2 = p \trans{p} \multiset{2q,p} &
t_3 = p \trans{q} \emptyset & t_4 = q \trans{q} \multiset{r,q} \\ 
t_5 = r \trans{p} \emptyset & t_6 = p \trans{q} \multiset{2q} &
t_7 = t_4 &  t_8 = q \trans{r} \emptyset
\end{array}$$
\noindent 
where $t=x \trans{y} m$ denotes that $x$ is the source place, $y$ the observed place, and $m$ the multiset of destination places of $t$. The firing sequence that realizes $H$ is $t_1 \, t_2^5 \, t_3^2 \, t_4 \, t_5 \, t_6^4 \, t_7 \, t_8^{18}$.
While the final marking of $H$ is produced by $\beta_{c}$ only,  $\beta_{c}$ is not realizable on its own. For example, the $r$-node
of $\beta_{c}$ at level 5  is destroyed in the next step by the firing of $t_5$, but $t_5$ can only occur if there is
at least one token in place $p$;  this token is supplied by $\beta_1$ or $\beta_2$. We can think of $\beta_1$ and $\beta_2$ as branching trajectories that eventually become extinct, but before extinction provide tokens that need to be observable to fire some transitions.
\end{example}

\subparagraph*{Cargo, fuel, and smoke of a history.} A \emph{decoration} $\dec{H}$ of a history $H$ 
consists of the history $H$ itself and a partition of the nodes of $H$ into \emph{cargo}, \emph{fuel}, and \emph{smoke}  
nodes. Figure \ref{fig:Bp(i)} shows not only a history $H$ but also a decoration $\dec{H}$.
Cargo nodes are white, grey nodes are fuel, and black nodes are smoke. Before giving the formal definition of
a decoration, let us provide some intuition. Think of the sequence of markings of a history as the sequence
of states of a ship. All nodes of the final marking are cargo, they are what the ship ``delivers'' in the end.
At any other marking, the cargo nodes are the ``causal predecessors'' of the final cargo nodes. Every decoration
has the same cargo nodes, they only differ in the partition of the other nodes into fuel and smoke. Intuitively,
a decoration reserves the right to use fuel nodes to fire transitions (a $p$-node can be ``used''  to fire  a transition that observes $p$), and commits to never using a smoke node or its descendants. 
The most conservative decoration (which always exists) is the one that declares all non-cargo nodes as fuel. 
Our first goal will be to show that every history has an equivalent \emph{fuel-efficient} history that delivers the same 
cargo but admits a low-fuel decoration. 

Formally, a \emph{decoration} of $H$ is a partition of the nodes of $H$ 
into \emph{cargo}, \emph{fuel}, and \emph{smoke} nodes satisfying the following conditions:
\begin{itemize}
\item A node of $H$ is a \emph{cargo node} if{}f it has at least one descendant in $H(l)$.
\item All descendants of smoke nodes are smoke nodes.
\item For every place $p$ and level $i$, if $H(i)$ contains smoke $p$-nodes, then it also contains fuel $p$-nodes.
(``No smoke without fuel''. Intuitively, the smoke $p$-nodes are not needed because the fuel $p$-nodes can be used instead.)
\end{itemize}
\noindent 
A \emph{decorated history} is a pair consisting of $H$ and a decoration of $H$. 
Observe that along all paths cargo comes before fuel, and fuel before smoke. Graphically, white nodes (if any) come before grey nodes (if any), and grey nodes before black nodes (if any).

%

\subparagraph*{Every history is equivalent to a fuel-efficient history.}

We prove that every realizable history has an equivalent realizable history with a \emph{fuel-efficient}
decoration, defined as follows:

\begin{definition}
Let $\dec{H}$ be a decorated history.
A place $p$ is \emph{wasteful at level $i$} if $\dec{H}(i)$ contains more than $n$ fuel $p$-nodes. A place $p$ is \emph{wasteful} in $\dec{H}$ if 
it is wasteful at some level; otherwise $p$ is \emph{fuel-efficient} in $\dec{H}$.  Finally, 
$\dec{H}$ is \emph{fuel-efficient} if all places are fuel-efficient.
\end{definition}

\begin{example}
Since $n=3$, in the decorated history of Figure \ref{fig:Bp(i)} place $p$ is wasteful at levels 1 to 6, and $q$ is wasteful at levels 3 to 8. 
The history is not fuel-efficient. 
\end{example}

The proof is based on a Replacement Lemma,
which plays the same role as the combination of the Pruning and Boosting Lemmas for IO nets. 
We start by introducing a definition.

\begin{definition}
The \emph{$(p, i)$-bunch} of $H$, denoted $B_p(i)$, is the set of subtrees of $H$ 
rooted at the $p$-nodes of $H(i)$. 
\end{definition}

Loosely speaking, the Replacement Lemma shows that if $i$ is the earliest level at which $p$ is wasteful, then
the bunch $B_p(i)$ of trajectories can be replaced so that the new history has a decoration where $p$ is not wasteful anymore.
The lemma shows how to do this while ensuring that the histories before and after the replacement are equivalent.
Repeated applications of the Replacement Lemma yield a fuel-efficient history. 

Formally, given a history $B'_p$ with $p$-nodes as roots and with the same number of trees as $B_p(i)$, we let $H[B'_p / B_p(i)]$ denote the result of replacing each tree of $B_p(i)$ by a different tree of $B'_p$. For this we assume that $B_p(i)$ and $B'_p$ have been enumerated in some way, and the $j$-th tree of $B_p(i)$ is replaced by the $j$-th tree of $B'_p$. We state the Replacement Lemma:

\begin{restatable}[Replacement Lemma]{lemma}{BIOPruningLm}
Let $\dec{H}$ be a decoration of a realizable history $H$ such that $p$ is wasteful, and $i$ is the earliest level 
at which $p$ is wasteful. There exists a history $B'_p$ such that $H'\defeq H[B_p'/B_p(i)]$  is realizable, equivalent to $H$, and has a decoration whose fuel-efficient places contain all fuel-efficient places of $\dec{H}$ and $p$. 
\end{restatable}
\begin{proof}(Sketch.)
We describe the history $B'_p$, illustrating the construction on the decorated history of Figure \ref{fig:Bp(i)}. In this example $p$ is already wasteful at level $i=1$, and $B_p(1) = H$. So all of $H$ is replaced by 
the bunch $B'_p$, shown in Figure \ref{fig:B'p}.

\begin{figure}[ht]
\centering
\resizebox{13cm}{!}{
\input{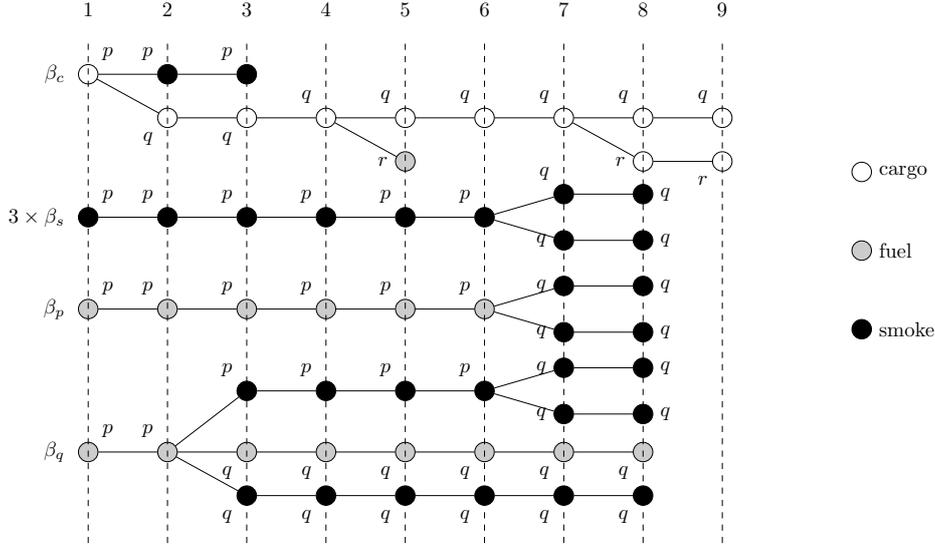}
}
\caption{Result of replacing $B_p(1)$ in the history of Figure \ref{fig:Bp(i)}.}
\label{fig:B'p}
\end{figure}

In order to describe $B'_p$ we need some notions. 
We call smoke and fuel nodes \emph{transportation} nodes.
Given a decorated history $\dec{H}$, 
let $\mathit{last}(p)$ denote the last level $i$ such that $\dec{H}(i)$ contains a transportation $p$-node.
A  \emph{place-level} is a pair $(q, j)$, where $q$ is a place and $j$ is a level of $H$.
A \emph{path} of place-levels is a concatenation of ``steps'' of two types:  ``doing nothing steps'' from $(r,l)$ to $(r,l+1)$ such that $l < last(r)$, 
and ``transportation steps'' from $(r,l)$ to $(s,l+1)$ such that some transportation $r$-node of $\dec{H}(l)$ has an $s$-child in $\dec{H}(l+1)$. 
We say that $(q, j)$ is \emph{reachable} from $(p,i)$ if there is a path from $(p,i)$ to $(q,j)$, and let $\bunchreach$ be the set of all place-levels $(q,j)$ reachable from $(p,i)$. 
In our example we have $\mathcal{R}_{p,1}=\{(p,1), \ldots, (p,6), (q,3), \ldots, (q, 8) \}$. (Observe that $(r, 5)$ does not belong to $\mathcal{R}_{p,1}$, because its parent is a cargo node.)

$B_p'$ is the union of three sets of branching trajectories, $B_c$, $B_f$, and $B_s$ (where $c, f, s$ stand for cargo, fuel, and smoke):
\begin{itemize}
\item $B_c$ contains all branching trajectories of $B_p(i)$ rooted at a cargo node.
(In Figure \ref{fig:B'p}, $B_c$ is the singleton set $\set{\beta_{c}}$.) 
The decoration of $B_c$ is chosen so that it conserves the cargo nodes of $\dec{H}$.
Intuitively, $B_c$ ensures that $H'$ delivers the same cargo as $H$. 
\item $B_f$ contains a branching trajectory $\beta_q$ for every $q$ such that $(q,j) \in \bunchreach$ for some $j$. 
(In Figure \ref{fig:B'p}, $B_f$ contains the two trees $\beta_p$ and $\beta_q$.) Intuitively, these trajectories guarantee that the new set $\bunchreach$ of $\dec{H'}$ is a superset of the old one, and so that any transition firing that relies on observing some place $q$ at level $j$ can still occur, because $(q,j)$ is still reachable from $(p, i)$. 

Let us now define $\beta_q$.  (Figure \ref{fig:betaq} shows $\beta_q$ for the history of Figure \ref{fig:B'p}.)
Let $\mathit{first}(q)$ be the smallest $j$ such that $(q, j) \in \bunchreach$. 
There is a shortest path from $(p,i)$ to $(q,\mathit{first}(q))$, and
each step of the path corresponds to doing nothing or to executing a transition once.  (In Figure \ref{fig:betaq} we have $(p,i)=(p,1)$, $(q,\mathit{first}(q))=(q,3)$, and the path corresponds to doing nothing in the first step, and then firing $t_2$.) Let $\delta_{q}$ be the corresponding branching trajectory. (In Figure \ref{fig:betaq}, $\delta_{q}$ is the tree contained in the blue area.) First we append a path to each leaf of $\delta_q$: If the leaf is, say, an $r$-node at level $j$, then we append to it a path of $r$-nodes from level $j$ to level $\mathit{last}(r)$. (Red area of Figure \ref{fig:betaq}.) 
Then, we append to 
the end of each path a \emph{destroyer}, i.e., a tree that makes the token disappear. 
We choose for this any subtree of $\dec{H}$ rooted in a transportation node of $(r, \mathit{last}(r)).$
(Green area of Figure \ref{fig:betaq}; in order to destroy a $p$-node we first transform it into two $q$-nodes by firing $t_6$, wait while $t_7$ is fired in another part of the history, and then destroy the $q$-nodes by firing $t_8$ twice. The two $q$-nodes are destroyed by firing $t_8$ twice.) 
The decoration of $\beta_q$ is chosen so that there is a fuel path rooted in $(p,i)$ containing $q$-nodes from levels $\mathit{first}(q)$ to $\mathit{last}(q)$, and the rest is smoke.

\item $B_s$ contains $|B_p(i)| - |B_c| - |B_f|$ copies of a tree of smoke nodes $\beta_{s}$, consisting of a path of $p$-nodes, leading from level $i$ to level $\mathit{last}(p)$, appended with a destroyer. Intuitively, this is smoke added to ensure that $H(i)=H'(i)$.
\end{itemize}

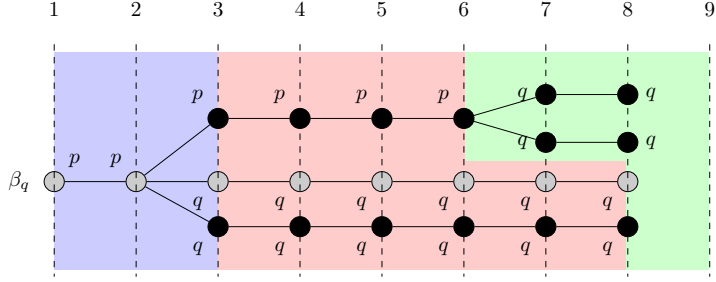
\begin{figure}[ht]
\resizebox{!}{4cm}{%
%

\newcommand{\dist}{*0.7}
\newcommand{\distbranchbig}{*0.4}
\newcommand{\distbranch}{*0.05}
\newcommand{\disttree}{*0.1}

\begin{tikzpicture}

\tikzstyle{bone}=[circle, draw,minimum size=2mm]
\tikzstyle{fat}=[circle, draw,minimum size=1mm,fill=black!20]
\tikzstyle{meat}=[circle, draw,minimum size=1mm,fill=black]
\tikzstyle{delta}=[circle, draw,minimum size=1mm,fill=blue!20]
\tikzstyle{branch}=[circle, draw,minimum size=1mm,fill=red!20]
\tikzstyle{destroy}=[circle, draw,minimum size=1mm,fill=green!20]

\draw[fill=blue!20, draw=none] (0.0, -0.3) rectangle (2.7, -3.9) {};
\draw[fill=red!20, draw=none] (2.7, -0.3) rectangle (6.75, -3.9) {};
\draw[fill=red!20, draw=none] (6.75, -2.1) rectangle (9.4, -3.9) {};
\draw[fill=green!20, draw=none] (6.75, -0.3) rectangle (10.75, -2.1) {};
\draw[fill=green!20, draw=none] (9.4, -2.1) rectangle (10.75, -3.9) {};

    \node[bone,draw=none]
    (Node01) [label=above:$1$]
    {};
    \node[below=3\dist of Node01, fat]
    (Node131) [label=above right:$p$]
    {};
    \node[below=2\dist of Node131, bone, draw=none]
    (Node151)
    {};

    \node[right=1 of Node131, fat]
    (Node132) [label=above left:$p$]
    {};
    \node[right=1 of Node132, fat]
    (Node133)  [label= below left:$q$]
    {};
    \node[right=1 of Node133, fat]
    (Node134)  [label= below left:$q$]
    {};
    \node[right=1 of Node134, fat]
    (Node135)  [label= below left:$q$]
    {};
    \node[right=1 of Node135, fat]
    (Node136)  [label= below left:$q$]
    {};
    \node[right=1 of Node136, fat]
    (Node137)  [label= below left:$q$]
    {};
    \node[right=1 of Node137, fat]
    (Node138)  [label= below left:$q$]
    {};
    \node[above=0.7 of Node133, meat]
    (Node113) [label=above left:$p$]
    {};
    \node[below=1\distbranchbig of Node133, meat]
    (Node143) [label=below left:$q$]
    {};
    \node[right=1 of Node113, meat]
    (Node114) [label=above left:$p$]
    {};
    \node[right=1 of Node114, meat]
    (Node115) [label=above left:$p$]
    {};
    \node[right=1 of Node115, meat]
    (Node116) [label=above left:$p$]
    {};
    \node[right=1 of Node116, bone, draw=none]
    (Node117)
    {};
    \node[above=1\distbranch of Node117, meat]
    (Node107) [label=left:$q$]
    {};
    \node[below=1\distbranch of Node117, meat]
    (Node127) [label=left:$q$]
    {};
    \node[right=1 of Node107, meat]
    (Node108) [label= right:$q$]
    {};
    \node[right=1 of Node127, meat]
    (Node128) [label= right:$q$]
    {};
    \node[right=1 of Node143, meat]
    (Node144) [label=below left:$q$]
    {};
    \node[right=1 of Node144, meat]
    (Node145) [label=below left:$q$]
    {};
    \node[right=1 of Node145, meat]
    (Node146) [label=below left:$q$]
    {};
    \node[right=1 of Node146, meat]
    (Node147) [label=below left:$q$]
    {};
    \node[right=1 of Node147, meat]
    (Node148) [label=below left:$q$]
    {};
    
    \node[right=1 of Node01, bone, draw=none]
    (Node02) [label=above:$2$]
    {};
    \node[right=1 of Node02, bone, draw=none]
    (Node03) [label=above:$3$]
    {};
    \node[right=1 of Node03, bone, draw=none]
    (Node04) [label=above:$4$]
    {};
    \node[right=1 of Node04, bone, draw=none]
    (Node05) [label=above:$5$]
    {};
    \node[right=1 of Node05, bone, draw=none]
    (Node06) [label=above:$6$]
    {};
    \node[right=1 of Node06, bone, draw=none]
    (Node07) [label=above:$7$]
    {};
    \node[right=1 of Node07, bone, draw=none]
    (Node08) [label=above:$8$]
    {};
    \node[right=1 of Node08, bone, draw=none]
    (Node09)  [label=above:$9$]
    {};
    \node[right=1 of Node151, bone, draw=none]
    (Node152)
    {};
    \node[right=1 of Node152, bone, draw=none]
    (Node153)
    {};
    \node[right=1 of Node153, bone, draw=none]
    (Node154)
    {};
    \node[right=1 of Node154, bone, draw=none]
    (Node155)
    {};
    \node[right=1 of Node155, bone, draw=none]
    (Node156)
    {};
    \node[right=1 of Node156, bone, draw=none]
    (Node157)
    {};
    \node[right=1 of Node157, bone, draw=none]
    (Node158)
    {};
    \node[right=1 of Node158, bone, draw=none]
    (Node159)
    {};
    
%

    
    \draw[-] (Node131) to
    (Node132);
    \draw[-] (Node132) to
    (Node133);
    \draw[-] (Node132) to
    (Node113);
    \draw[-] (Node132) to
    (Node143);
    \draw[-] (Node133) to
    (Node134);
    \draw[-] (Node134) to
    (Node135);
    \draw[-] (Node135) to
    (Node136);
    \draw[-] (Node136) to
    (Node137);
    \draw[-] (Node137) to
    (Node138);
    \draw[-] (Node113) to
    (Node114);
    \draw[-] (Node114) to
    (Node115);
    \draw[-] (Node115) to
    (Node116);
    \draw[-] (Node116) to
    (Node107);
    \draw[-] (Node116) to
    (Node127);
    \draw[-] (Node107) to
    (Node108);
    \draw[-] (Node127) to
    (Node128);
    \draw[-] (Node143) to
    (Node144);
    \draw[-] (Node144) to
    (Node145);
    \draw[-] (Node145) to
    (Node146);
    \draw[-] (Node146) to
    (Node147);
    \draw[-] (Node147) to
    (Node148);
    
    \draw[-, dashed] (Node01) to
    (Node151);
    \draw[-, dashed] (Node02) to
    (Node152);
    \draw[-, dashed] (Node03) to
    (Node153);
    \draw[-, dashed] (Node04) to
    (Node154);
    \draw[-, dashed] (Node05) to
    (Node155);
    \draw[-, dashed] (Node06) to
    (Node156);
    \draw[-, dashed] (Node07) to
    (Node157);
    \draw[-, dashed] (Node08) to
    (Node158);
    \draw[-, dashed] (Node09) to
    (Node159);

        \node[left=1\disttree of Node131](Tree4){$\beta_{q}$};
        
\end{tikzpicture}
}
\caption{Illustration of the construction of the set $B_f$ of trees.}
\label{fig:betaq}
\end{figure}

This concludes the description of $B'_p$. 
There are at most $|B_f| \le n$ fuel nodes per level in $B'_p$, so $p$ is fuel-efficient. 
The proof that $H'$ is realizable, equivalent to $H$, and has a decoration in which there are no new wasteful places can be found in  the appendix.
\end{proof}

Repeated applications of the Replacement Lemma yield the existence of a fuel-efficient decoration $\dec{H'}$ of a history $H'$
equivalent to $H$.

%
\begin{example}
Applying the Replacement Lemma to $p$ and $i:=1$ and the decorated history $\dec{H}$ of Figure \ref{fig:Bp(i)} 
yields the decorated history $\dec{H'}$ of Figure \ref{fig:B'p}.
Like $H$, it leads from $(6,0,0)$ to $(0,1,1)$. 
It is realized by $t_1 \, t_2 \, t_3 \, t_4 \, t_5 \, t_6^5 \, t_7 \, t_8^{12}$.
Place $p$ is no longer wasteful in $\dec{H'}$, and in fact all places are fuel-efficient.
\end{example}

The next step of the proof is the Unique Footprint Lemma. Loosely speaking,
it shows that for every history there exists an equivalent history in which 
any two levels differ in the cargo, the fuel, or the \emph{support} of the smoke. This allows us to bound the 
length of the history. 
We need a preliminary lemma. Let $\dec{H}_c(i)$, $\dec{H}_f(i)$, $\dec{H}_s(i)$ denote the multisets of cargo, fuel, and smoke
nodes of $\dec{H}(i)$. Intuitively, the Smoke Irrelevance lemma shows that we can always deliver the same
cargo using the same fuel \emph{independently} of the initial amount of smoke.

\begin{restatable}[Smoke Irrelevance Lemma]{lemma}{SmokeIrrelevance}
\label{lm:blue-irrelevance}
Let $\dec{H}$ be a realizable decorated history of length $h$, and let $\mu$ be any 
multiset of places such that $\support{\mu} \subseteq  \support{\dec{H}_s(1)}$.
There exists a realizable decorated history $\dec{{H'}}$ of length $h$ such that 
 $\dec{H'}_s(1) = \mu$, and $\dec{H'}_c(i) = \dec{H}_c(i)$  and $\dec{H'}_f(i) = \dec{H}_f(i)$ for every level $1 \leq i \leq h$.
\end{restatable}
\begin{proof}(Sketch.) Rename $\nu \defeq\dec{H}_s(1)$ for clarity. To construct $\dec{{H'}}$,
start with $\dec{H}$, and do the following for every place $p \in P$. If $\mu(p) \le \nu(p)$, then delete $\nu(p) - \mu(p)$ smoke $p$-nodes from $\dec{H}(1)$ as well as all their descendants (which are all smoke nodes by definition). If $\mu(p) > \nu(p)$, then add to $\dec{H}$ $(\mu(p) - \nu(p))$ copies of an arbitrary tree $\beta$ of smoke nodes of $\dec{H}$ rooted in $(p,1)$. This tree exists because $p \in \support{\mu}$, and so $p \in \support{\nu}$. The addition of the copies of $\beta$ maintains the ``no smoke without fuel'' property, because it was already fulfilled in $\dec{H}$ by the nodes of $\beta$. The smoke nodes of $\dec{H}'(1)$ thus constructed are labelled by $\mu$, and fuel and cargo nodes are neither added nor removed. The proof that $\dec{H}'$ is realizable can be found in the Appendix.
\end{proof}

\begin{definition}
Given a level $\dec{H}(i)$ of a decorated history, define its \emph{footprint} as the triple
$(\dec{H}_c(i), \dec{H}_f(i), \support{\dec{H}_s(i)})$ (that is, we only take the support of
$\dec{H}_s(i)$, not $\dec{H}_s(i)$ itself).
\end{definition}

\begin{lemma}[Unique Footprint Lemma]
\label{lm:unique-footprint}
Every realizable history has an equivalent fuel-efficient decorated history in which every level has a different fooprint.
\end{lemma}
\begin{proof}
Let $\dec{H}$ be a realizable decorated history. By the Replacement Lemma, we can assume 
w.l.o.g. that $\dec{H}$ is fuel-efficient. Assume further that $\dec{H}$ has minimal length $h$, i.e., every
equivalent decorated history that is also fuel-efficient has length at least $h$. We claim that
every level of $\dec{H}$ has a different footprint. 
Assume this is not the case. Then there exist two indices $1 \le i < j \le h$ such that 
$(\dec{H}_c(i), \dec{H}_f(i), \support{\dec{H}_s(i)}) = (\dec{H}_c(j), \dec{H}_f(j), \support{\dec{H}_s(j)})$. 
The truncated history $\dec{H}(j)\dec{H}(j+1) \ldots \dec{H}(h)$ is clearly realizable. Since 
$\support{\dec{H}_s(i)} = \support{\dec{H}_s(j)}$, we can apply the Smoke Irrelevance Lemma with
$\mu:= \dec{H}_s(i)$ and obtain a decorated history $\dec{H'}$ 
of length $h-j+1$
such that $(\dec{H}_c(i), \dec{H}_f(i), \dec{H}_s(i)) = (\dec{H'}_c(1), \dec{H'}_f(1), \dec{H'}_s(1))$
(notice: now $\dec{H}_s(i)=\dec{H'}_s(1)$, instead of only $\support{\dec{H}_s(i)}=\support{\dec{H'}_s(1)}$ ).
But this implies $\dec{H}(i) = \dec{H'}(1)$, and so the concatenation $H(1) \cdots H(i-1) H'(1) \cdots H'(h-j+1)$ is also
a realizable history. By the Smoke Irrelevance Lemma we have $\dec{H'}_c(h-j+1) = \dec{H}_c(h)$. Since the last 
levels of a decorated history only contain cargo nodes, this implies $\dec{H'}(h-j+1) = \dec{H}(h)$, and so 
the concatenation is equivalent to $H$. Further, since $\dec{H'}$ has the same cargo and fuel nodes as 
$\dec{H}(j)\dec{H}(j+1) \ldots \dec{H}(h)$, the concatenation is also fuel-efficient, contradicting that
$\dec{H}$ has minimal length. 
\end{proof}

We are equipped to prove the Shortening Theorem.

\BimoIntermediateValues*
\begin{proof} 
We first prove the bound on the accelerated length.
By the Unique Footprint Lemma, there is a history $H$ such that $\sourceMarking \trans{H} \targetMarking$
and $H$ has a decoration $\dec{H}$ where every level has a different footprint. So the length of 
$\dec{H}$ is bounded by the number of possible footprints of the histories leading from
$\sourceMarking$ to $\targetMarking$. Since, by definition, the number of cargo nodes cannot decrease from a level to the next, and the last level consists of only cargo, every level has between $0$ and $m$ cargo nodes per place. Since $\dec{H}$ is fuel-efficient, every level has between $0$ and $n$ fuel nodes per place. Finally, there are at most $2^n$ possible supports in a net with $n$ places. So the number of footprints, and so the length of $\dec{H}$, and the accelerated length of any firing sequence realizing $\dec{H}$, is at most $2^n(m+1)^n(n+1)^n$.

Let us now prove the token bound. 
To bound the number of smoke nodes in each level, we apply the following operation. 
Replace every largest tree of smoke nodes (since the children of smoke nodes are smoke, this means trees rooted at smoke nodes whose parents are cargo or fuel) by the tree $\beta_{s}$ defined as in the Replacement Lemma: 
$\beta_{s}$ is a path of smoke $p$-nodes ending at level $\mathit{last}(p)$, appended by a $p$-destroyer tree. 
This maintains realizability, because (by the ``no smoke without fuel'' property in $\dec{H}$),  it does not decrease the support of the multiset of places of any level. 
We call $\dec{H'}$ the resulting realizable history with decorated nodes.
Note that the ``no smoke without fuel'' property may not hold in $\dec{H'}$,  so it is not formally a decorated history, but it is sufficient to conclude the proof.
$\dec{H'}$ has the following property: smoke $p$-nodes can only create other nodes (which, by definition, are also smoke) at the level $\mathit{last}(p)$, and it can create at most $m_d$ of them. 

At all other levels $j$ of $\dec{H}'$, only cargo and fuel nodes can create nodes.
There are at most $h' \leq 2^n(m+1)^n(n+1)^n$ levels, and at most $(m+n)$ cargo and fuel nodes per place.
Each transition has a unique source place,
and
all the nodes are added to the initial $m'$ nodes corresponding to the tokens of $\sourceMarking$.
Thus there are at most $m'+ h' (m+n) m_d$ nodes at the first level $last(p)$ in which a smoke node creates nodes.
At most all of the nodes are smoke, 
so at most $(m'+ h' (m+n) m_d)m_d$ nodes are created.
There are at most $n$ levels $\mathit{last}(p)$, 
which each create at most the total amount of nodes times $m_d$ nodes.
Thus at every level of the history there are at most $(m' + h' (m+n) m_d )m_d^{n}$ nodes,
concluding the proof.

\end{proof}

\section{Many-to-many reachability and coverability.}
\label{sec:param-reach-BIO}
In~\cite{conf/apn/EsparzaRW19} we prove that many-to-many versions of the reachability and coverability problems for IO nets are $\PSPACE$-complete. 
We extend this result to \bio{} nets, which requires to use not only the Shortening Theorem itself, but also the lemmas 
conducting to its proof.  

We recall some definitions of \cite{conf/apn/EsparzaRW19}.
A set $\cube$ of markings of a net $\net=(P,T,F)$  is a \emph{cube} if there exist mappings 
$L \colon P \rightarrow \N$ and $U \colon P \rightarrow \N \cup \infty$ such that $M \in \cube$ if and only if $L\leq M \leq U$. Abusing language,
we identify $\cube$ with the pair $(L, U)$. Observe that cubes can be infinite sets
of markings.  The \emph{cube-reachability (coverability)} consists of deciding, given a net $\net$ and cubes $\cube,\cube'$ of $\N$, whether there exist markings 
$M \in \cube$ and $M' \in \cube'$ such that $M$ is reachable (coverable) from $M'$.

\begin{restatable}{theorem}{TheoremBioCube}
The cube-reachability and cube-coverability problems for \bio{} nets are \PSPACE-complete.
\end{restatable}
\begin{proof}
\PSPACE-hardness follows from \PSPACE-hardness for IO nets. 
We show that the problems are in \NPSPACE\ and apply Savitch's theorem. 
Cube-coverability from $\cube'$ to $\cube=(L,U)$ reduces to cube-reachability from $\cube'$ to the cube $(L,U'')$ such that $U''(p)=\infty$ for all $p$, so it
suffices to consider cube-reachability from $\cube'=(L',U')$ to $\cube=(L,U)$. 
%
For each place $p$ with upper bound $U(p)=\infty$ in $\cube$, add a ``destroying transition'' $\tau_p$ to $N$
with preset $\preset{\tau_p}=\set{p}$ and postset $\postset{\tau_p}=\emptyset$. 
We guess a marking $M$ of size $m$ satisfying $M(p)=L(p)$ if $U(p)=\infty$, and $L(p) \leq M(p) \leq U(p)$ if $U(p) < \infty$. 
This reduces the problem to checking if $M$ is reachable
in the modified net
from some marking of $\cube'$.
%
%
%
By Lemma \ref{lm:blue-irrelevance}, only the footprint of a marking matters for knowing whether it can reach marking $\targetMarking$.
We pick $\sourceMarking$ in $\cube'$ of size $m' \le m+n^2+\max(|L'|,n)$.
 The summands correspond to the cargo, fuel  
and smoke nodes of the initial marking
 of a fuel-efficient decorated history given by the Replacement Lemma if $\sourceMarking \trans{*} \targetMarking$ holds,
 where $\max(|L'|,n)$ is enough smoke nodes so that $\sourceMarking \in \cube'$ and any set of places is covered.
By Theorem \ref{thm:simple-BIO-reachability}, $\sourceMarking \trans{*} \targetMarking$ can be checked in \PSPACE. 
\end{proof}


\section{Conclusion}
\label{sec:conclusion}
We have shown that immediate observation Petri nets are globally flat,
allowing the use of existing efficient verification tools.
We have also studied \bioLong{} nets, which are simultaneously a generalisation
of IO nets, and of the Basic Parallel Processes model.
The class of \bio{} nets 
significantly extends the expressive power 
of both IO nets and BPP nets,
bringing together process creation and (restricted) cross-process interaction
via a simple and natural definition.
While such an extension does not preserve global flatness,
we have proven that local flatness is still preserved,
and many-to-many reachability and coverability problems are still in \PSPACE{}.

As \bio{} nets combine \PSPACE{}-verifiable reachability and non-semilinear
reachability relation, the further study of the structure of this reachability 
relation seems of interest. For instance, we plan to obtain the bounds on the size
of the pre- and post- image of a marking, provided that these images are finite.
It is also worth noting that the results of this paper still hold
(up to a slight alteration of the Shortening Theorem bounds)
if we define BIO transitions via the constraint $|\preset{t}-\postset{t}| \le 1$.
This is equivalent to extending BIO transitions with the possibility of 
multiple observations and the absence of a source place.




\bibliography{references}

\appendix
\section{Shortening Theorem for IO nets}

\BoostingLemma*
\begin{proof} (Sketch.) Let $h$ be the length of $H$, and
let $t_1^{k_1} \cdots t_{h-1}^{k_{h-1}}$ be a realization of $H$. 
For every $1 \leq i \leq h-1$ define $k_i'$ as follows: if $\tau(i) = \tau(i+1)$, then $k_i' \defeq k_i$;
if $\tau(i) \neq \tau(i+1)$, then $k_i' \defeq k_i +1$. 
We claim that 
$t_1^{k_1'} \cdots t_{h-1}^{k_{h-1}'}$ is a realization of  $H+\tau$. 
The proof is by induction on $h$.

Assume $h=1$.  Then $H$ is realizable by $t^0$ for any transition $t$, and so is $H+\tau$. 

Assume that the induction property holds for some $h \ge 1$, and let $H$ be of length $h+1$, realizable by $t_1^{k_1} \cdots t_{h}^{k_{h}}$.
By induction, the history $H+\tau$ truncated of its last step is realizable by $t_1^{k_1'} \cdots t_{h-1}^{k_{h-1}'}$.
If $\tau(h) \ne \tau(h+1)$ in $H$, then since $\tau \in H$ and $H$ is realizable, 
$\tau(h)\tau(h+1)=p_s p_d$ for $p_s$ and $p_d$ the source and destination places of $t_{h}$.
Additionally, 
there are $k_{h}-1$ other trajectories $\tau'$ such that $\tau'(h)\tau'(h+1)=p_s p_d$, and there is at least one trajectory $\tau'$ such that $\tau'(h)\tau'(h+1)=p_o p_o$.
Thus $t_1^{k_1'} \cdots t_{h-1}^{k_{h-1}'} t_h^{k_h+1}$ realizes $H+\tau$.
If $\tau(h)=\tau(h+1)$ in $H$, then $H+\tau$ 
is realized by $t_1^{k_1'} \cdots t_{h-1}^{k_{h-1}'} t_h^{k_h}$. 
\end{proof}

\ThmShortIO*
\begin{proof}
Let $H$ be a realizable history such that $\sourceMarking \trans{H} \targetMarking$, and let $h$ be the length of $H$.
For every two places $p, q$, let $B_{p,q}$ denote the bunch of all trajectories of $H$ leading from $p$ to $q$,
and let $T_{p,q} = \size{B_{p,q}}$.  Applying the Pruning Lemma to all bunches $B_{p,q}$ such that $T_{p,q} \geq n$, 
we obtain a new realizable history $\widetilde{H}$ satisfying
\begin{equation}
\label{eqT}
\widetilde{T}_{p,q} = \min\{ n, T_{p,q} \} \quad \mbox{ for every $p, q \in P$.}
\end{equation}
So $\widetilde{H}$ has $\sum_{p,q \in P} \widetilde{T}_{p,q} \leq n^3$ trajectories. 
Let $M_{\widetilde{H}}^1 \trans{t_1^{k_1}} M_{\widetilde{H}}^2 \cdots  M_{\widetilde{H}}^{h-1} \trans{t_{h-1}^{k_{h-1}}} M_{\widetilde{H}}^h$ be a realization of $\widetilde{H}$. Since $\widetilde{H}$ hast at most $n^3$ trajectories,
we have $M_{\widetilde{H}}^i(p) \leq n^3$ for every $p \in P$ and $1 \leq i \leq n$. If $h \geq (n^3 +1)^n$,
then there are $1 \leq i \neq j \leq h$ such that $M_{\widetilde{H}}^i = M_{\widetilde{H}}^j$, and 
the history $\widetilde{H'}$ obtained by ``cutting out'' the fragment of $\widetilde{H}$ between 
$M_{\widetilde{H}}^i$ and $M_{\widetilde{H}}^j$ is also realizable. (Formally, $\widetilde{H'}$ 
is the result of replacing every trajectory $\tau \in \widetilde{H}$ by $\tau(1) \cdots \tau(i)\tau(j+1) \cdots \tau(h)$.)
So w.l.o.g. we can assume $\widetilde{h} <  (n^3 +1)^n$.

Since $\widetilde{H}$ is realizable, we have $\widetilde{\sourceMarking} \trans{\widetilde{H}} \widetilde{\targetMarking}$
for some markings $\widetilde{\sourceMarking}, \widetilde{\targetMarking}$. We examine the relation between 
$\sourceMarking$ and $\widetilde{\sourceMarking}$, and between $\targetMarking$ and $\widetilde{\targetMarking}$.
For every place $p$, the initial (final) number of tokens of $p$ in $H$ is equal to the number of trajectories
of $H$ of  starting in $p$ (ending in $p$), and similarly for $\widetilde{H}$. So we have  
$$\begin{array}{rcl}
\sourceMarking(p) = \sum_{q \in P} T_{p,q} & \mbox{ and } & \targetMarking(p) = \sum_{q \in P} T_{q,p} \\
\widetilde{\sourceMarking}(p) = \sum_{q \in P} \widetilde{T}_{p,q} & \mbox{ and } &
\widetilde{\targetMarking}(p) = \sum_{q \in P} \widetilde{T}_{q,p}
\ . 
\end{array}$$ 
\noindent Further, for every place $p \in P$:
\begin{itemize}
\item[(a)] $\widetilde{\sourceMarking}(p) \leq \sourceMarking(p)$, and $\widetilde{\targetMarking}(p) \leq \targetMarking(p)$. \\
Follows immediately from  $\widetilde{T}_{p,q} \leq T_{p,q}$ for every $q \in P$ (Equation \ref{eqT}).
\item[(b)] If $\widetilde{\sourceMarking}(p) = 0$ then $\sourceMarking(p) = 0$, and  if $\widetilde{\targetMarking}(p) = 0$ then 
$\targetMarking(p) = 0$. \\
If $\widetilde{\sourceMarking}(p) = 0$ then $\widetilde{T}_{p,q}=0$ for every $q \in P$. So,  by Equation \ref{eqT},
$\widetilde{T}_{p,q}=T_{p,q}$ for every $q \in P$, and so $\sourceMarking(p) = \sum_{q \in P} T_{p,q} =
\sum_{q \in P} \widetilde{T}_{p,q} = \widetilde{\sourceMarking}(p) = 0$. The proof for the target markings is analogous.
\end{itemize}

Let $\overline{H}$ be the history obtained from $\widetilde{H}$ as follows: For every $p, q \in P$, if $\widetilde{T}_{p,q}>0$
then pick a trajectory $\tau \in B_{p,q}$, and set $\overline{B}_{p,q} = \widetilde{B}_{p,q} + (\widetilde{T}_{p,q} - T_{p,q} -1) \cdot \tau$
By the Boosting Lemma, $\overline{H}$ is realizable, and so there are markings $\overline{\sourceMarking}, \overline{\targetMarking}$ such that
$\overline{\sourceMarking} \trans{\overline{H}} \overline{\targetMarking}$. Further, by (a) and (b) above we have 
$\overline{T}_{p,q} = T_{p,q}$ for every $p, q \in P$, and so for every $p \in P$:
$$\overline{\sourceMarking}(p) = \sum_{q \in P} \overline{T}_{p,q} = \sum_{q \in P} T_{p,q} = \sourceMarking(p)$$
So we get $\sourceMarking \trans{\overline{H}} \targetMarking$. Since $\widetilde{H}$ and $\overline{H}$ have the same length, 
we get $\overline{h} <  (n^3 +1)^n$. So every firing sequence realizing $\overline{H}$ has accelerated length at most 
$(n^3 +1)^n$, and we are done.
\end{proof}

\section{Shortening Theorem for BIO nets}
%
%
%

We give ourselves a few more definitions to help in the proofs.
We call smoke and fuel nodes \emph{transportation} nodes.
Given a decorated history $\dec{H}$, 
let $\mathit{last}(p)$ denote the last level $i$ such that $\dec{H}(i)$ contains a transportation $p$-node.
A  \emph{place-level} is a pair $(q, j)$, where $q$ is a place and $j$ is a level of $H$.
A \emph{path} of place-levels is a concatenation of ``steps'' of two types:  ``doing nothing steps'' from $(r,l)$ to $(r,l+1)$, 
and ``transportation steps'' from $(r,l)$ to $(s,l+1)$ such that some transportation $r$-node of $\dec{H}(l)$ that has an $s$-child in $\dec{H}(l+1)$. 
We say that $(q, j)$ is \emph{reachable} from $(p,i)$ if there is a path from $(p,i)$ to $(q,j)$, and let $\bunchreach$ be the set of all place-levels $(q,j)$ reachable from $(p,i)$.


\BIOPruningLm*
\begin{proof}
We first construct $H'\defeq H[B_p'/B_p(i)]$ and show that it is realizable and equivalent to $H$.
Then, we define a decoration $\dec{H'}$ of $H'$, and show that it realizes the condition of the lemma.

 \textbf{Construction of $H'$.}
We define $B_p'$ as the union of three sets of branching trajectories, $B_c$, $B_f$, and $B_s$ (where $c, f, s$ stand for cargo, fuel, and smoke):
\begin{itemize}
\item $B_c$ contains all branching trajectories of $B_p(i)$ rooted at a cargo node. 
\item $B_f$ contains a branching trajectory $\beta_q$ for every $q \in \bunchreach$. 

We define $\beta_q$.
Let $\mathit{first}(q)$ be the smallest $j$ such that $(q, j) \in \bunchreach$. 
Notice that $\mathit{first}(q) \le \mathit{last}(q)$ for all $q \in P$, since by definition of reachability there exists a transportation $q$-node in level $\mathit{first}(q)$.
There is a shortest path from $(p,i)$ to $(q,\mathit{first}(q))$, and
each step of the path corresponds to doing nothing or to executing a transition once. 
Let $\delta_{q}$ be the corresponding branching trajectory.
First we append a path to each leaf of $\delta_q$: If the leaf is, say, an $r$-node at level $j$, then we append to it a path of $r$-nodes from level $j$ to level $\mathit{last}(r)$. 
Then, we append to 
the end of each path a \emph{destroyer}, i.e., a tree that makes the token disappear. 
We choose for this any subtree $\gamma_r$ of $\dec{H}$ rooted in a transportation node of $(r, \mathit{last}(r)).$
\item $B_s$ contains $|B_p(i)| - |B_c| - |B_f|$ copies of a tree $\beta_{s}$, consisting of a path of $p$-nodes, leading from level $i$ to level $\mathit{last}(p)$, appended with a destroyer $\gamma_p$. 
\end{itemize}

We define the replacement $H' = H[B_p'/B_p(i)]$:
we replace the trees of $B_p(i)$ with a cargo root in $\dec{H}$ by the same tree in $B_c$,
we replace some trees of $B_p(i)$ with a fuel root in $\dec{H}$ by the trees of $B_f$ (in any order),
and the rest  of the trees of $B_p(i)$ by the trees of $B_s$. 
This is well-defined because the 
the trees of $B'_p$ all have $p$-nodes as root, 
there are no more than $n$ trees in $B_f$ and more than $n$ trees with fuel roots in $B_p(i)$ since $p$ is wasteful at $i$,
and there are as many trees overall in $B'_p$ as in $B_p(i)$.

\medbreak \textbf{History $H'$ is realizable and equivalent.}
History $H'$ is equivalent to history $H$: 
the trees added in $B'_p \setminus B_c$ all end in destroyers,
and the other trees of $H'$ were already in $H$,
so $H'$ has the same final marking.
In case $i=1$, 
the number of $p$-nodes in $H(i)$ and $H'(i)$ is the same so $H'$ has the same initial marking.

History $H$ is realizable, and we note $t_1^{k_1} \cdots t_{h-1}^{k_{h-1}}$ a sequence that realizes it, for some transitions $t_1, \ldots, t_{h-1} \in T$ and numbers $k_1, \ldots, k_{h-1} \geq 0$.
We show that $H'$ is realizable using the same transitions but different numbers $l_1, \ldots, l_{h-1} \geq 0$.
Let $1 \le j \le h-1$. 
Let $H'_p(i)$ be the set of nodes of $H'(j)$ which have exactly one child with the same label, and let $H'_a(j)$ be the rest.
We claim that for every node $v'$ in $H'_a(j)$ with label $r$ and multiset of children labels $c$, there exists a node $v$ in $H_a(j)$ with label $r$ and multiset of children labels $c$.
By realizability of $H$ this entails that $v'$ is labeled with the source place $p_s$ of $t_j$, and the multiset of labels of its children is the multiset $\multiset{p_{d_1}, \ldots, p_{d_k}}$ of destinations of $t_j$.

Now to show our claim.
Let $v'$ a node of $H'_a(j)$. 
If $v'$ is not a node of the subtree $B'_p$, or if $v'$ is a node of $B_c$, then we are done.
Let us assume this is not the case, \ie $v' \in B_f \cup B_s$.
\begin{itemize}
\item
If $v'$ is in a tree $\beta_{s}$, then it is in a a destroyer (since $v'$ is not in $H'_p(j)$) and so it is in a copy of a subtree of $\dec{H}$ .
\item
Assume $v'$ is in a tree $\beta_q$ for some $q \in \bunchreach$.
If $v'$ is in a destroyer then it is in a copy of a subtree of $\dec{H}$, we are done.
Otherwise, $v'$ is in the tree $\delta_{q}$ induced by the shortest path $\shortpath_q$ from $(p,i)$ to $(q,\mathit{first}(q))$ in $H$.
Since $v'$ is not passive, \ie $v' \notin H'_p(j)$, and by definition of how a path induces a tree,
there is an $r$-node $v$ of $H(j)$ with the same children as $v'$.
\end{itemize}

We now show that the set $H'_p(j)$ contains a node labeled by the place $p_o$ observed by $t_j$.
If there is a node labeled $p_o$ in $H_p(j)$ that is not in $B_p(i)$, then it is also in $H'_p(j)$ and we are done.
Let us assume that the only nodes of $H_p(j)$ labeled $p_o$ are in $B_p(i)$.
If there is a cargo node labeled $p_o$ in $\dec{H}_p(j)$ then it is also in $H'_p(j)$ so we are done.
Otherwise there exists a transportation node $v$ labeled $p_o$ in $\dec{H}_p(j)$, and $j \le \mathit{last}(p_o)$ by definition.
Since $v$ is in $B_p(i)$, either $v$ is in a tree with a cargo root, or place-level $(p_o,j)$ is reachable from $(p,i)$.
If $v$ is in a tree of $B_p(i)$ with a cargo root, it is also in $B_c\subseteq B'_p$.
Otherwise $(p_o,j)\in \bunchreach$, and therefore by construction there is a node in $B'_p$ labeled $p_o$ at every level between $\mathit{first}(p_o)$ and $\mathit{last}(p_o)$, in particular at $j$.

\medbreak \textbf{Decoration of $H'$.}
Let $\dec{H'}$ be the following decoration of $H'$.

We start with the nodes of $B_f$ and $B_s$.
In each tree $\beta_q$ in $B_f$, 
constructed around the tree induced by a shortest path $\shortpath_q$ from $(p,i)$ to $(q, \mathit{first}(q))$,
we let the nodes along the path $\shortpath_q$ be fuel nodes, 
along with the nodes along one branch from $(q, \mathit{first}(q))$ to $(q, \mathit{last}(q))$.
All the other nodes of $\beta_q$ are defined as smoke nodes.
We let all the nodes of the trees $\beta_{s}$ be smoke nodes.

The rest of the nodes of $H'$ are decorated in two steps.
First, we set $\dec{H'}$ to be equal to $\dec{H}$ on the nodes of $H' \setminus (B_f \cup B_s)$, 
which is possible because $H' \setminus B'_p = H \setminus B_p(i)$ and the trajectories of $B_c$ are trajectories of $B_p(i)$.
Then, we do the following ``re-decoration''.
Let $(q,j)$ be a place level reachable from $(p,i)$ in $H'$.
If there are any fuel nodes labeled $q$ in $(H' \setminus B_f)(j)$, redecorate them and all their descendants as smoke nodes in $\dec{H'}$.
Do this for every $(q,j)$ reachable from $(p,i)$.

The order of ``cargo then fuel then smoke'' is respected along the branching trajectories of $\dec{H'}$
because
they are respected in $B'_p$, 
and there are no more than $n$ trees with a fuel root in $B'_p$ while there are more than $n$ in $B_p(i)$.
The cargo nodes in $\dec{H'}$ are well defined, as the cargo nodes of $\dec{H'}$ are the cargo nodes of $\dec{H}$. 

The smoke/fuel partition of $\dec{H'}$ is well defined: 
First, remark that the last level index $\mathit{last}(p)$ at which there is a transportation $p$-node in $\dec{H'}$ is equal to $\mathit{last}(p)$ in $\dec{H}$, for any place $p$ by construction. 
Let $v$ be a smoke $q$-node at level $\dec{H'}(j)$, for some $q$ and $j$.
We check that the ``no smoke without fuel'' condition is fulfilled.
If $(q,j)$ is reachable from $(p,i)$ in $H$ then there exists a fuel $q$-node in $\dec{H'}(j)$ provided by $\beta_q$, since $j\le \mathit{last}(q)$ by virtue of $v$ being smoke. 
If $(q,j)$ is not reachable from $(p,i)$ in $H$,
then there is no subtree of $B_p(i)$ rooted in a transportation $p$-node with a descendant labeled $q$.
Therefore in $\dec{H'}$, node $v$ is not in $B_f$. 
Since it is also not in $B_s$, whose trees are only $p$ nodes until $\mathit{last}(p)$, $v$ is in either a tree of $B_c$ or in no tree of $B'_p$,
and therefore $v$ exists also in $\dec{H}$ as a smoke node.
Since the smoke/fuel partition of $\dec{H}$ is well defined, there exists a fuel $q$-node $v'$ in $\dec{H}(j)$.
Since $(q,j)$ is not reachable from $(p,i)$ in $H$, $v'$ is either in $B_c$ or not part of $B_p(i)$ and so $v'$ is also in $\dec{H'}(j)$.

For every place-level $(q,j)$ in $H'$ reachable from $(p,i)$, there are at most $n$ fuel $q$-nodes in $\dec{H'}(j)$.
Indeed, by definition, the only fuel nodes labeled $q$ in $\dec{H'}(j)$ are in $B_f(j)$.
By definition of $B_f(j)$, the only fuel nodes labeled $q$ in $B_f(j)$ are in the trees $\beta_r$ for some $r \in \bunchreach$.
There are a most $n$ such trees, and in each tree there is at most one fuel node per level.

Therefore there are no wasteful places $q$ at some level $j$ such that 
$(q,j)$ is reachable from $(p,i)$ in $H'$.
In particular, $p$ is fuel-efficient since $i$ is the earliest level at which $p$ is wasteful in $H$. 
If there is a wasteful place-level in $\dec{H'}$, then it is unreachable from $(p,i)$ in $H'$.
By definition of $H'$, this means that it is also a wasteful place-level in $H \setminus B'_p$ and thus in $H$. 
Thus the fuel-efficient places of $\dec{H'}$ contain all the fuel-efficient places of $\dec{H}$, as well as the place $p$.
\end{proof}


We remind the reader that $\dec{H}_c(i)$, $\dec{H}_f(i)$, $\dec{H}_s(i)$ denote the multiset of cargo, fuel, and smoke
nodes of $\dec{H}(i)$. 

\SmokeIrrelevance*
%
\begin{proof}
Rename $\nu:=\dec{H}_s(1)$ for clarity. 
To construct $\dec{{H'}}$,
start with $\dec{H}$, and do the following for every place $p \in P$. 
If $\mu(p) \le \nu(p)$, then delete $\nu(p) - \mu(p)$ smoke $p$-nodes from $\dec{H}(1)$ as well as all their descendants (which are all smoke nodes by definition). 
If $\mu(p) > \nu(p)$, then add to $\dec{H}$ $(\mu(p) - \nu(p))$ copies of an arbitrary tree $\beta$ of smoke nodes of $\dec{H}$ rooted in $(p,1)$. 
This tree exists because $p \in \support{\mu}$, and so $p \in \support{\nu}$. 
The addition of the copies of $\beta$ maintains the ``no smoke without fuel'' property, because it was already fulfilled in $\dec{H}$ by the nodes of $\beta$.

The smoke nodes of $\dec{H}'(1)$ thus constructed are labelled by $\mu$, and fuel and cargo nodes are neither added nor removed.
We prove that $\dec{H}'$ is realizable.
Let $t_1^{k_1} \cdots t_{h-1}^{k_{h-1}}$ be a sequence that realizes $\dec{H}$, for some transitions $t_1, \ldots, t_{h-1} \in T$ and numbers $k_1, \ldots, k_{h-1} \geq 0$.

Removing trees of smoke nodes from $\dec{H}'$ does not affect realizability:
if there is a smoke $p_o$-node labeled by the observed place of $t_i$ in some level $\dec{H}(i)$ with a child labeled the same in $\dec{H}(i+1)$, then there is also a  pair of such fuel $p_o$-node in $\dec{H}(i)$ and $\dec{H}(i+1)$ by property of smoke nodes.
This pair of fuel nodes is still in $\dec{H}'$ because we only remove trees of smoke nodes.
Removing the trees translates as decreasing the iterations of some transitions in the realizing sequence of $\dec{H}$.
The trees of smoke nodes that we add to $\dec{H}'$ also do not affect realizability: 
they only increase the iterations of the transitions in the realizing sequence, 
as in the proof of the Replacement Theorem.
\end{proof}

\end{document}